\documentclass[10pt, twocolumn, comsoc]{IEEEtran}

\usepackage{graphicx,epsfig}
\usepackage[noadjust]{cite}
\usepackage{mcite}
\usepackage{amsfonts,helvet}
\usepackage{fancyhdr}
\usepackage{threeparttable}
\usepackage{epsf,epsfig}
\usepackage{amsthm}
\usepackage{amsmath}
\usepackage{siunitx}
\usepackage{amssymb}
\usepackage{booktabs}
\usepackage{dsfont}
\usepackage{subfigure}
\usepackage{color}
\usepackage[linesnumbered,ruled,noend]{algorithm2e}
\usepackage{algpseudocode}
\usepackage{algcompatible}
\usepackage{enumerate}
\usepackage{gensymb}
\usepackage{cancel}
\usepackage{bbm}
\usepackage{graphicx,subfigure}
\usepackage{graphicx}
\usepackage{array}
\usepackage{mathtools}
\usepackage{textcomp}
\usepackage{siunitx}
\newcolumntype{P}[1]{>{\centering\arraybackslash}p{#1}}
\usepackage{stfloats} 

\newtheorem{lemma}{Lemma}

\usepackage{eucal}

\def\bF{{\bf F}}

\makeatletter
\def\blfootnote{\xdef\@thefnmark{}\@footnotetext}
\makeatother

\begin{document}

\title{
    Power-Constrained and Quantized MIMO-RSMA Systems
    with Imperfect CSIT: Joint Precoding, Antenna Selection, and Power Control
}

\author{
    Jiwon Sung, {\it Student Member, IEEE},
    Seokjun Park, {\it Student Member, IEEE},
    and Jinseok Choi, {\it Member, IEEE}

    \thanks{
        J. Sung, S. Park, and J. Choi are with the School of Electrical Engineering, Korea Advanced Institute of Science and Technology, Daejeon, South Korea (e-mail: {\texttt{\{jiwonsung, sj.park, jinseok\}@kaist.ac.kr}}).
        
    }
}

\maketitle \setcounter{page}{1} 

\begin{abstract}
    To utilize the full potential of the available power at a base station (BS), we propose a joint precoding, antenna selection, and transmit power control algorithm for a total power budget at the BS.
    We formulate a sum spectral efficiency (SE) maximization problem for downlink multi-user multiple-input multiple-output (MIMO) rate-splitting multiple access (RSMA) systems with arbitrary-resolution digital-to-analog converters (DACs).
    We reformulate the problem by defining the ergodic sum SE using the conditional average rate approach to handle imperfect channel state information at the transmitter (CSIT), and by using approximation techniques to make the problem more tractable.
    Then, we decompose the problem into precoding direction and power control subproblems.
    We solve the precoding direction subproblem by identifying a superior Lagrangian stationary point, and the power control subproblem using gradient descent.
    We also propose a complexity-reduction approach that is more suitable for massive MIMO systems.
    Simulation results not only validate the proposed algorithm but also reveal that when utilizing the full potential of the power budget at the BS, medium-resolution DACs with $8 \sim 11$ bits may actually be more power-efficient than low-resolution DACs.
\end{abstract}


\blfootnote{This work was presented in part at the {\em IEEE Veh. Technol. Conf. 2025 (to appear)} \cite{sung2025joint_TOAPPEAR}.}

\section{Introduction}

Multiple-input multiple-output (MIMO) \cite{mimo:twc:10:marzetta} architectures can obtain spectral efficiency (SE) gains by simply increasing the number of antenna elements at a base station (BS). 
However, practical deployments of large antenna arrays face a non-trivial circuit power consumption challenge when the BS has limited power resources.
In order to utilize the full potential of the available power at the BS and maximize the sum SE of users, the BS must strategically balance the trade-off between the transmit power and the circuit power consumption instead of employing a lot of antennas with very low transmit power.
To lower the circuit power, one possible solution is to selectively turn off antennas that contribute little to the overall communication performance.
Such antenna selection strategies have been widely studied in the literature \cite{sanayei:commmag:04, tervo2015selection, vlachos2020energy, choi2021energyIOTJ}.
Meanwhile, another possible solution is to reduce the number of quantization bits of quantizers such as digital-to-analog converters (DACs), which are primary power consumers at the transmitter.
As the power consumption of a DAC is exponentially proportional to the number of quantization bits \cite{zhang2018low}, decreasing the DAC resolution can significantly reduce the power consumption. 
However, this comes at the cost of having non-negligible quantization errors.
This greatly deteriorates the accuracy of channel estimation, which can cause severe inter-user interference (IUI).
In this regard, employing a rate-splitting multiple access (RSMA) scheme \cite{mao2018rate_bri}, which can effectively mitigate IUI and overcome the limitations on SE gains in downlink multi-user multiple-input multiple-output (MU-MIMO) systems \cite{kaleva:tsp:16, joudeh:16:tsp, dai:twc:16, li:jsac:20,clerckx2016rate}, can play a key role in MIMO systems with low-resolution quantizers.

As the precoder, the set of active antennas, and the transmit power are all coupled to each other, it is necessary to design an algorithm that jointly optimizes these three aspects in order to utilize the full potential of the available power at a BS and achieve a better SE performance.
In this paper, we propose such a joint optimization algorithm for downlink MIMO-RSMA systems with arbitrary-resolution DACs that maximizes the sum SE of users given a total power budget at the BS.
In addition, we consider imperfect channel state information at the transmitter (CSIT).
Lastly, we also propose a complexity-reduction approach of the proposed algorithm that is more suitable for massive MIMO systems.

\subsection{Prior Works}

As mentioned earlier, one way to reduce the circuit power consumption at a BS is to use low-resolution DACs, however, at the cost of introducing significant nonlinear quantization errors that negatively affect the SE performance.
To handle such nonlinear errors in a tractable way, linear approximation methods such as the Bussgang decomposition \cite{mezghani2012capacity} and the additive quantization noise model (AQNM) \cite{orhan2015low} are commonly used for analysis.
Such linear approximations were also widely adopted in quantized MIMO communications \cite{fan2015uplink, jacobsson2017quantized, zhang2018mixed, ribeiro2018energy, ding2019spectral, park2022rate, choi2021energyIOTJ}.

Some works considered low-resolution quantizers for precoding design.
In \cite{jacobsson2017quantized}, conventional linear precoding methods such as the minimum mean square error (MMSE) or zero-forcing (ZF) were proposed using low-resolution DACs with $3 \sim 4$ bits.
This approach was further refined by using an alternating
one-bit precoding method in \cite{chen2018alternating}.
Additionally in \cite{wang2018finite}, the alternating direction method of multipliers (ADMM) was adopted to solve an IUI minimization problem.
Furthermore in \cite{choi2021energyIOTJ}, an algorithm based on the generalized power iteration (GPI) method was proposed to solve an energy efficiency (EE) maximization problem.

Another strategy to lower the circuit power consumption at a BS is to use antenna selection algorithms.
In \cite{sanayei:commmag:04}, a simple norm-based antenna selection algorithm was proposed where antennas with the lowest channel gain are turned off first.
Subsequently, in \cite{tervo2015selection}, a joint beamforming and antenna selection algorithm for maximizing the EE under perfect CSIT was proposed.
In \cite{vlachos2020energy}, low-resolution quantization was considered in its EE maximizing antenna selection algorithm under perfect CSIT; however, antenna selection was performed while fixing the precoders.
Lastly, in \cite{choi2021energyIOTJ}, a joint beamforming and antenna selection algorithm for an EE maximization problem using low-resolution quantizers was proposed under perfect CSIT.

Various approaches have been developed to enhance the SE performance using RSMA with 
imperfect CSIT \cite{joudeh2016sum, li:jsac:20, park2021rate}, quantization error \cite{park2022rate}, and secure communications \cite{fu2020robust}.
Although these prior works on RSMA frameworks provided state-of-the-art solutions that maximize the SE, they set the BS to use the maximum transmit power and did not consider the circuit power consumption in their optimization.
Using the maximum transmit power does not guarantee the optimal spectral efficiency performance in scenarios where the BS has limited power resources, as it may be better to lower the transmit power and use the extra power to utilize more antennas.
In contrast, a comprehensive optimization that balances the trade-off between the transmit power and the circuit power consumption allows for a strategic balancing of the antenna deployment and the transmit power based on the system requirements.
For instance, the BS may activate more antennas with reduced transmit power, or conversely, use fewer antennas with increased transmit power, depending on the specific system requirements and constraints.

Meanwhile, the works that maximize the EE using antenna selection while considering the circuit power consumption \cite{tervo2015selection, vlachos2020energy, choi2021energyIOTJ} did not explicitly optimize the SE, and thus, cannot guarantee a desirable SE performance, which limits their applicability in practice.
In addition, they also did not consider imperfect CSIT.
In this regard, it is of interest to consider the total power budget for the joint optimization of antenna selection, precoding, and transmit power control in MIMO-RSMA systems under coarse quantization to accomplish efficient and maintainable communication systems.

\subsection{Contributions}
In this paper, we develop a joint precoding, antenna selection, and transmit power control method for maximizing the sum SE in a downlink quantized MIMO-RSMA system under a total power budget constraint at the BS.
Our main contributions are summarized below:
\begin{itemize}
    \item
    We formulate the joint optimization problem, which
    stands apart from most prior approaches \cite{joudeh2016sum, li:jsac:20, park2021rate, park2022rate,tervo2015selection, vlachos2020energy, choi2021energyIOTJ}
    in that it explicitly considers the trade-off between the transmit power and the circuit power consumption within the total power budget.
    
    \item 
    We consider imperfect CSIT, where the covariance matrix of the channel estimation error is assumed to be known, and reformulate the problem to leverage the error covariance matrix in the optimization.
    To this end, we define the ergodic sum SE using the conditional average achievable rate approach \cite{joudeh2016sum}
    and establish a lower bound for the sum SE given both the estimated channel and the channel error covariance matrix.
    We then incorporate the lower bound into the objective function.

    \item 
    We propose a joint precoding, antenna selection, and power control optimization algorithm given a total power budget.
    We first use the Lagrangian method to incorporate the power constraint into the objective function.
    Then, the lower bound-based problem is divided into two subproblems: one for optimizing the precoding direction with joint antenna selection and the other for optimizing the transmit power.
    We solve the former using a GPI-based approach and the latter using gradient descent.
    Lastly, we adopt the bisection method to efficiently find the Lagrange multiplier.

    \item 
    We further propose a complexity-reduction approach of the proposed algorithm that is more suitable for massive MIMO systems.
    First, we assume independent and identically distributed (IID) channels to approximate the covariance matrix of the channel estimation error as a diagonal matrix.
    Then, we recursively compute matrix inversions using the Sherman-Morrison formula to significantly reduce the complexity.
    
    \item 
    We numerically show that the proposed algorithm outperforms other baselines in terms of the sum SE performance and the complexity.
    One of the key takeaways of our work is that medium-resolution DACs can actually be more power-efficient than low-resolution DACs when utilizing the full potential of the power resources at the BS.
\end{itemize}

\textit{Notations}:
$\text{a}$, ${\bf{a}}$, and ${\bf{A}}$ are a scalar, a vector, and a matrix.
The superscripts $(\cdot)^{\sf T}$, $(\cdot)^{\sf H}$, and $(\cdot)^{-1}$ denote the matrix transpose, Hermitian, and inversion, while $\mathbb{E}[\cdot]$, ${\rm Tr}(\cdot)$, and $\otimes$ represent the expectation, the trace, and the Kronecker product operators, respectively.
Lastly, ${\bf 0}_{N}$ and ${\bf{I}}_N$ denote the $N \times 1$ zero vector and $N \times N$ identity matrix, respectively.

\section{System Model}
\label{sec:sys_model}

\subsection{Signal Model}

We consider a downlink MU-MIMO system where a BS equipped with $N$ antennas serves $K$ single-antenna users.
We assume arbitrary-bit DACs at the BS.
Our system employs the single-layer rate-splitting (RS) architecture \cite{clerckx2016rate,joudeh2016sum} as our transmission scheme, as illustrated in Fig.~\ref{fig:RS-SE-DAC-ADC system model}.
Each individual message $W_k$ of user $k$ is initially split into a common part $W_k^{\sf c}$ and a private part $W_k^{\sf p}$.
Following this, the common parts of all users are combined and encoded to create a single common message, i.e., $W^{\sf c} = \left( W_1^{\sf c}, \cdots, W_K^{\sf c} \right)$.
The resulting common message $W^{\sf c}$ is then encoded into a common stream $s_{\sf c}$ using a public codebook so that $s_{\sf c}$ can be decodable later on by all $K$ users in the network.
Meanwhile, a private codebook is used to encode the private messages $W_1^{\sf p}, \cdots, W_K^{\sf p}$ into private streams $s_1, \cdots, s_K$ so that they are decodable only by their respective users.
Upon receiving a signal, a user first decodes the common stream $s_{\sf c}$, removes it via SIC, then decodes the corresponding private stream $s_k$.
Here, we assume no error propagation when performing SIC, and also assume that the common and private streams follow a Gaussian distribution with zero mean and unit variance, i.e., $s_{\sf c}, s_k \sim {\mathcal {CN} (0,1)}$ $\forall k$.

\begin{figure}[!t]\centering
	\subfigure{\resizebox{1\columnwidth}{!}{\includegraphics{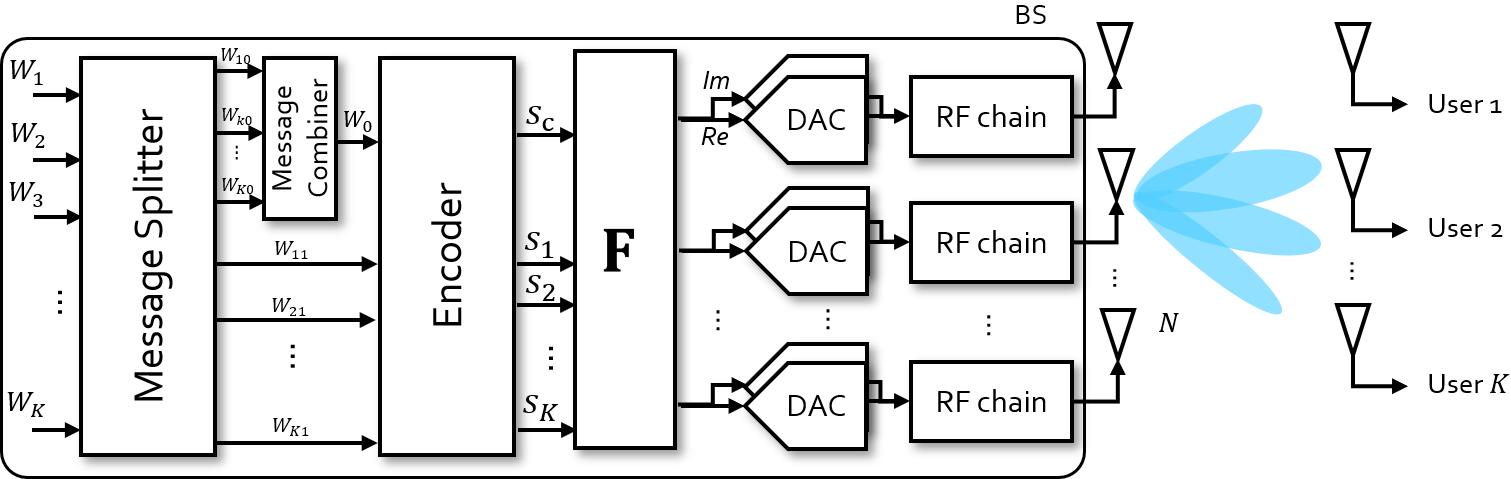}}}
	\caption{A visualized downlink MU-MIMO RSMA system.} 
 	\label{fig:RS-SE-DAC-ADC system model}
 	\vspace{-1em}
\end{figure}

Let $P$ denote the maximum transmit power at the BS.
The precoded digital baseband transmit signal ${\bf{x}} \in \mathbb{C}^{N}$ is 
\begin{align}
    \label{eq:digital base band signal}
    {\bf{x}} = \sqrt{P}{\bf{f}}_{0} s_{\sf c} + \sqrt{P}\sum_{k = 1}^{K} {\bf{f}}_k s_k,
\end{align}
where ${\bf f}_{0} \in \mathbb{C}^{N}$ and ${\bf f}_k \in \mathbb{C}^{N}$ are the precoding vectors for the common and private streams, respectively.
Then, the digital baseband signal ${\bf{x}}$ in \eqref{eq:digital base band signal} undergoes quantization at the DACs before being transmitted.
Using the AQNM technique \cite{fletcher2007robust}, which linearly approximates the quantization process under the assumption of the MMSE quantizer, the analog baseband transmit signal after undergoing quantization is expressed as
\begin{align} 
    \label{eq:linear_quantization_model}
     {Q}({\bf{x}}) \approx   {\bf x}_{\sf q} 
     & = {\bf{\Phi}}_{{\alpha }_{\sf DAC}} {\bf x} + {\bf q}_{\sf DAC} \\
     & =  \sqrt{P}{\bf{\Phi}}_{{\alpha }_{\sf DAC}} {\bf{f}}_{0} s_{\sf c} + \sqrt{P}{\bf{\Phi}}_{{\alpha }_{\sf DAC}}\sum_{k = 1}^{K} {\bf{f}}_k s_k + {\bf q}_{\sf DAC},
\end{align}
where $Q(\cdot)$, ${\bf{\Phi}}_{{\alpha }_{\sf DAC}} = {\rm diag}({\alpha }_{{\sf DAC},1},\ldots,{\alpha }_{{\sf DAC},N})\in \mathbb{R}^{N \times N}$, and ${\bf q }_{\sf DAC}\in \mathbb{C}^{N}$
denote an element-wise quantizer, a quantization loss matrix, and a quantization noise vector, respectively.
The quantization loss of the $i$th DAC, ${\alpha }_{{\sf DAC}, i} \in (0, 1)$, is defined as ${\alpha }_{{\sf DAC}, i} = 1 - {\beta }_{{\sf DAC}, i}$, where ${\beta }_{{\sf DAC}, i}$ is the normalized mean squared quantization error, i.e., ${ \beta }_{\sf DAC,\textit{i}} = \frac {\mathbb{E}{[| {x_i - x_{{\sf q}, i}}|^2]}} {\mathbb{E}{[|{x_i}|^2]}}$ \cite{fletcher2007robust,zhang2018mixed}.
Here, $x_i$ and $x_{{\sf q}, i}$ are the $i$th element of ${\bf x}$ and ${\bf x}_{\sf q}$, respectively, and the values of ${\beta }_{{\sf DAC}, i}$ vary based on the number of quantization bits of the $i$th antenna, $b_{{\sf DAC}, i}$.
Table 1 in \cite{fan2015uplink} depicts the values of ${\beta }_{{\sf DAC}, i}$ for $b_{{\sf DAC}, i} \leq 5$.
For $b_{{\sf DAC}, i} \geq 6$, ${\beta }_{{\sf DAC}, i}$ can be approximated as ${\beta }_{{\sf DAC}, i} \approx \frac {\pi \sqrt{3}}2 {2^{-2b_{{\sf DAC},i}}}$ \cite{gersho2012vector}.
In addition, ${{\bf q}_{\sf DAC}}$ is uncorrelated with the digital baseband signal $\bf{x}$  and follows ${{\bf q}_{\sf DAC}} \backsim {\mathcal {CN} ({\bf 0}_{N},\bf R_{{\bf q}_{\sf DAC}})}$ \cite{fletcher2007robust}.
Since the MMSE quantizer satisfies $\mathbb{E} {[ x_i | x_{{\sf q}, i} ]} = x_{{\sf q}, i}$ \cite{fletcher2007robust}, we can use the linear model in \eqref{eq:linear_quantization_model} to derive the covariance matrix of ${\bf q}_{\sf DAC}$:
\begin{align} 
    {\bf R}_{{\bf q}_{\sf DAC}}
    \nonumber
    = &\; {\bf{\Phi}}_{{\alpha }_{\sf DAC}}{\bf{\Phi}}_{{\beta }_{\sf DAC}}{\rm diag} \left({\mathbb E}{\left[\bf x\bf x^{\sf H}\right]}\right) \\
    \label{eq:covariance matrix of quantization noise2}
    = &\; {\bf{\Phi}}_{{\alpha }_{\sf DAC}}{\bf{\Phi}}_{{\beta }_{\sf DAC}}{\rm diag} \left(P{\bf F \bf F^{\sf H}}\right),
\end{align}
where ${\bf{\Phi}}_{{\beta}_{\sf DAC}} = {\rm diag}({\beta }_{{\sf DAC},{\sf 1}},\ldots,{\beta }_{{\sf DAC},{N}})\in \mathbb{R}^{N\times N}$ and ${\bf F} = \left[{\bf f}_{0},{\bf{f}}_1, \cdots, {\bf{f}}_K\right]\in \mathbb{C}^{N\times(K+1)}$.

Subsequently, the received  baseband signal is modeled as
\begin{align} \label{eq:received signal}
    {\bf y} = {\bf H}^{\sf H}{\bf x}_{\sf q} + {\bf n},
\end{align}
where ${\bf H}^{\sf H} \in \mathbb {C}^{K \times N}$ is a downlink channel matrix between the BS and $K$ users, and ${{\bf n}} \backsim {\mathcal {CN} ({\bf 0}_{K},\sigma^{2}{\bf I}_{K})}$ is an additive white Gaussian noise (AWGN) with zero mean and variance of ${\sigma^2}$.
The received signal of user $k$ is
\begin{align}
    \nonumber
    y_{k}
    = &\; \sqrt{P}{\bf h}_k^{\sf H}{\bf{\Phi}}_{{\alpha }_{\sf DAC}}{\bf{f}}_{0} s_{\sf c}
    + \sqrt{P}{\bf h}_k^{\sf H}{\bf{\Phi}}_{{\alpha }_{\sf DAC}}{\bf{f}}_{k} s_{k} \\
    \label{eq:recevied signal at user k}
    & + \sqrt{P} \sum_{i =1, i \neq k}^{K} {\bf h}_k^{\sf H}{\bf{\Phi}}_{{\alpha }_{\sf DAC}}{\bf{f}}_i s_i
    + {\bf h}_k^{\sf H}{\bf q }_{\sf DAC}+{n}_k,
\end{align}
where ${\bf h}_k$ denotes the $k$th column of the channel matrix $\bf H$ and $n_k$ denotes the $k$th element of ${\bf n}$.

We express the covariance matrix of user $k$'s channel as ${\bf R}_{{\bf h}_k} = \mathbb{E}[{\bf h}_k {\bf h}_k^{\sf H}]$.
Let the estimated channel of user $k$ be denoted as $\hat{{\bf h}}_k = {\bf h}_k - \Tilde{{\bf h}}_k$, where $\Tilde{{\bf h}}_k$ is the channel estimation error vector.
For the case of frequency division duplex (FDD) MU-MIMO systems, the accuracy of the CSIT error is mainly determined by the amount of feedback bits to quantize the downlink channel.
Then, considering the FDD system, $\hat{{\bf h}}_k$ is modeled as follows \cite{wagner2012large, choi2019GPI}:
\begin{align} \label{eq:channelvector}
    {\hat {\bf{h}}}_{k} = \sqrt{1-\kappa^2}{\bf h}_k + {\bf q}_k.
\end{align}
Here, ${\bf q}_k$ denotes the CSIT error of user $k$  \cite{wagner2012large, choi2019GPI}, and $\kappa \in [0, 1]$ is a parameter that indicates the channel quality.
The covariance matrix of the channel estimation error vector $\Tilde{{\bf h}}_k$ is subsequently derived as
\begin{align} \label{eq:channelcov_R1}
    {\bf{R}}_{\Tilde{{\bf h}}_k}
    = \mathbb{E}{\left[ \Tilde{{\bf h}}_k \Tilde{{\bf h}}_k^{\sf H} \right]} 
    = {\bf U}_k {\bf \Lambda}_k^{1/2}(2 -2\sqrt{1-\kappa^2}){\bf \Lambda}_k^{1/2}{\bf U}_k^{\sf H},
\end{align}
where ${\bf \Lambda}_k$ is the diagonal matrix composed of the non-zero eigenvalues of ${\bf{R}}_{{\bf h}_k}$, and ${\bf U}_k$ is the matrix composed of the corresponding eigenvectors.
We assume that the BS has knowledge of only the estimated channel $\hat{\bf h}_k$ and the channel estimation error covariance matrix ${\bf{R}}_{\Tilde{{\bf h}}_k}$ for all users.
We remark that although our proposed algorithm indeed works for both FDD and time division duplex (TDD) systems, we focus on the FDD system since a similar approach can be used for TDD systems.

\subsection{Performance Metrics and Problem Formulation}
According to the RSMA decoding principle \cite{joudeh2016sum},
the decodability of the common stream by all users must be guaranteed to ensure successful SIC; consequently, the code rate of the common stream $s_{\sf c}$ is set as the minimum rate among all users.
Let $R_{\sf c}({\bf F}) = \min_{k} \{R_{{\sf c},k}({\bf F}) \}$ be the SE
of $s_{\sf c}$. Then,
\begin{align}
    \label{eq:ergodic common message}
    R_{\sf c}({\bf F})
    = \min_{k} \left\{  \log_2 \left(1 + \frac{P|{\bf h}_k^{\sf H}{\bf{\Phi}}_{{\alpha }_{\sf DAC}}{\bf{f}}_{0}|^2} {\text{IUI}_{\text c}({\bf F}) + \text{QE}_{k}({\bf F}) + {\sigma}^2} \right)\right\},
\end{align}
where 
\begin{align}
    \text{IUI}_{\text c}({\bf F}) = &\; P {\sum_{i = 1}^{K} |{\bf{h}}_{k}^{\sf H}{\bf{\Phi}}_{{\alpha }_{\sf DAC}} {\bf{f}}_{i}|^2}, \\
    \label{eq:QEk}
    \text{QE}_{k}({\bf F}) = &\; {\bf h}_k^{\sf H}{\bf R}_{{\bf q}_{\sf DAC}}{\bf h}_k.
\end{align}
Let $R_k({\bf F})$ be the SE of $s_k$. With the common stream eliminated via SIC, $R_k({\bf F})$ can be written as
\begin{align}
    \label{eq:ergodic private message}
    R_{k}({\bf F}) &= \ \log_2 \left(1 + \frac{P|{\bf h}_k^{\sf H}{\bf{\Phi}}_{{\alpha }_{\sf DAC}}{\bf{f}}_{ k}|^2} {\text{IUI}_{k}({\bf F}) + \text{QE}_{k}({\bf F}) + {\sigma}^2} \right),
\end{align}
where $\text{IUI}_{k}({\bf F}) = P{\sum_{i = 1,i \neq k}^{K} |{\bf{h}}_{k}^{\sf H}{\bf{\Phi}}_{{\alpha }_{\sf DAC}} {\bf{f}}_{ i}|^2}$. Note how the interference is reduced due to the cancellation of the common stream.
Based on  \eqref{eq:ergodic common message} and \eqref{eq:ergodic private message}, the sum SE can be expressed as
\begin{align}
    { R}_\Sigma({\bf F}) = R_{\sf c}({\bf F}) + \sum_{k=1}^{K}{ R}_{k}({\bf F}).
\end{align}

We now turn our attention to the power constraint at the BS.
Let $P_{\sf tot}$ denote the total power that is allocated to the BS.
The power consumption at the BS must not exceed $P_{\sf tot}$, which concerns both the circuit power consumption and transmit power consumption.
That is, the sum of the transmit power consumption $\varsigma^{-1} P_{\sf tx}$ and the circuit power consumption $P_{\sf cir}$ should satisfy $\varsigma^{-1} P_{\sf tx} + P_{\sf cir} \leq P_{\sf tot}$.
Here, $\varsigma$ denotes the power amplifier (PA) efficiency, and $P_{\sf tx}$ is $P_{\sf tx} = 
{\text {Tr}}( \mathbb{E}{[{\bf x}_{\sf q}{\bf x}_{\sf q}^{\sf H}]} )$.

Before deriving the formula for $P_{\sf cir}$, we first define the power consumption of a DAC and an RF chain as shown in Fig.~\ref{fig:RS-SE-DAC-ADC system model}.
According to \cite{ribeiro2018energy, cui2005energy}, the DAC power consumption $P_{\sf DAC}$ (Watts) can be expressed as
\begin{align}
    P_{\sf DAC}(b_{{\sf DAC}, i}, f_s) = 1.5 \! \cdot \! 10^{-5} 2^{b_{{\sf DAC}, i}} + 9 \! \cdot \! 10^{-12}  f_s  b_{{\sf DAC}, i},
\end{align}
where $f_s$ denotes the sampling rate.
Meanwhile, if we let $P_{\sf LP}$, $P_{\sf M}$, $P_{\sf LO}$, $P_{\sf H}$, and $P_{\sf PA}$ denote the power consumption of the low-pass filter, mixer, local oscillator, $90^\circ$ hybrid with buffer, and power amplifier, respectively, the RF power consumption $P_{\sf RF}$ can be expressed as $P_{\sf RF} = 2P_{\sf LP} + 2P_{\sf M} + P_{\sf H}$.
Subsequently, the power consumption of the $i$th antenna is
\begin{align}
    P_{{\sf ant}, i} = 2P_{\sf DAC}(b_{{\sf DAC}, i}, f_s) + P_{\sf RF},
\end{align}
and the total power consumption of all the antennas is
\begin{align}
    \label{antenna power}
    P_{\sf ant}(\mathcal{A}) =
    {\sum_{i=1}^{N}}\mathbbm{1}_{\left\{i \in \mathcal{A} \right\}}
    P_{{\sf ant}, i},
\end{align}
where $\mathcal{A}$ is a set of active antennas, and $\mathbbm{1}_{\left\{E\right\}}$ is an indicator function that returns $\mathbbm{1}_{\left\{E\right\}} = 1$ when the event $E$ is true and $\mathbbm{1}_{\left\{E\right\}} = 0$ otherwise.
Finally, the circuit power consumption $P_{\sf cir}(\mathcal{A})$ is formulated as
\begin{align}
    \label{circuit power}
    P_{\sf cir}(\mathcal{A})
    = P_{\sf LO} + 
    P_{\sf ant}(\mathcal{A}).
\end{align}

Based on the RSMA sum SE and power consumption model, the optimization problem can then be formulated as
\begin{subequations} \label{eq:problem_main}
    \vspace{-5mm}
    \begin{align}
        \mathop{{\text{maximize}}}_{\mathcal{A},{\bf{F}}}& \;\;
        R_\Sigma({\bf F}(\mathcal{A})) \\
        \label{eq:problem_main_constraint}
        \mathop{{\text{subject to}}}& \;\; \varsigma^{-1} P_{\sf tx}({\bf F}(\mathcal{A})) + P_{\sf cir}(\mathcal{A}) \leq P_{\sf tot}, \\
        & \;\; P_{\sf tx}({\bf F}(\mathcal{A})) \leq P.
    \end{align}
\end{subequations}
Key challenges of this problem include the non-smoothness issue of the $\min\{\cdot\}$ function in \eqref{eq:ergodic common message} and the indicator function in \eqref{antenna power}, as well as the non-convexity of the problem with respect to the precoder.
In the next section, we tackle these challenges and propose an efficient algorithm to solve the  problem.

\section{Proposed Algorithm} \label{sec:main}

In this section, we propose a joint  precoding, antenna selection, and power control method to solve the problem in \eqref{eq:problem_main} by resolving the challenges.
To this end, we first derive a lower bound of the ergodic sum SE and further reformulate the bound by using a smooth approximation and an indicator approximation. 
Then, a low-complexity power iteration-based method is developed for joint antenna selection and precoding, and a gradient descent approach is adopted for power control.

\subsection{Conditional Average Rate Approach}

To leverage the channel error covariance matrix, we convert the  problem  in \eqref{eq:problem_main} using the conditional average rate approach \cite{joudeh2016sum}.
We begin by introducing the ergodic SEs:
\begin{align}
    \label{eq:ergodic SE common}
    {\bar R}_{{\sf c},k}({\bf F}) = &\; \mathbb{E}_{{\hat{\bf h}}_k}{\left[\mathbb{E}_{{\bf h}|{\hat{\bf h}}_k}{\left[ R_{{\sf c},k}({\bf F};{\hat{\bf{h}}}_k)\middle|  {\hat{\bf{h}}}_k\right]}\right]},
    \\
    {\bar R_{k}}({\bf F}) = &\; \mathbb{E}_{{\hat{\bf h}}_k}{\left[\mathbb{E}_{{\bf h}|{\hat{\bf h}}_k}{\left[ R_k({\bf F};{\hat{\bf{h}}}_k)\middle|  {\hat{\bf{h}}}_k\right]}\right]}.
\end{align}
Then, we aim to derive proper SE expressions as functions of the estimated channels and error covariance matrices by handling the expectations.

We subsequently rewrite the term related to the DAC quantization error covariance in \eqref{eq:QEk} to obtain a form that is more suitable for the GPI algorithm \cite{choi2019GPI}:
\begin{align}
    \nonumber
    \text{QE}_{k}({\bf F}) = &\; {\bf h}_k^{\sf H}{\bf R}_{{\bf q}_{\sf DAC}}{\bf h}_k \\
    \nonumber
    = &\; {\bf h}_k^{\sf H}{\bf{\Phi}}_{{\alpha }_{\sf DAC}}{\bf{\Phi}}_{{\beta }_{\sf DAC}} {\rm diag} \left( P \sum_{i = 0}^{K} {\bf{f}}_{i} {\bf{f}}_{i}^{\sf H} \right) {\bf h}_k \\
    = &\; P \sum_{i = 0}^{K} {\bf{f}}_{i}^{\sf H} {\bf{\Phi}}_{{\alpha }_{\sf DAC}}{\bf{\Phi}}_{{\beta }_{\sf DAC}} {\rm diag} \left( {\bf h}_k {\bf h}_k^{\sf H} \right) {\bf{f}}_{i}.
\end{align}
Now, if we incorporate the estimated channel $\hat{{\bf h}}_k = {\bf h}_k - \Tilde{{\bf h}}_k$ instead of the true channel ${\bf{h}}_k$ in \eqref{eq:recevied signal at user k}, we have
\begin{align}
    \nonumber
    y_{k}
    = &\; \sqrt{P}\left( \hat{{\bf h}}_k + \Tilde{{\bf h}}_k \right)^{\sf H}{\bf{\Phi}}_{{\alpha }_{\sf DAC}}{\bf{f}}_{0} s_{\sf c} + \sqrt{P}\left( \hat{{\bf h}}_k + \Tilde{{\bf h}}_k \right)^{\sf H}{\bf{\Phi}}_{{\alpha }_{\sf DAC}}{\bf{f}}_{k} s_{k} \\
    \nonumber
    & + \sqrt{P}\!\!\!\sum_{i =1, i \neq k}^{K} \!\!\! \left( \hat{{\bf h}}_k + \Tilde{{\bf h}}_k \right)^{\sf H}{\bf{\Phi}}_{{\alpha }_{\sf DAC}}{\bf{f}}_i s_i
    + \left( \hat{{\bf h}}_k + \Tilde{{\bf h}}_k \right)^{\sf H}{\bf q }_{\sf DAC}+{n}_k \\
    = &\; \sqrt{P} \hat{{\bf h}}_k^{\sf H}{\bf{\Phi}}_{{\alpha }_{\sf DAC}}{\bf{f}}_{0} s_{\sf c} + v_k.
\end{align}
Since the CSIT error $\Tilde{{\bf h}}_k$ is uncorrelated with the desired signal under MMSE channel estimation \cite{hassibi2003much},  $v_k$ is treated as noise uncorrelated with the desired signal. 
Then further considering it as Gaussian noise results in a lower bound on the mutual information expression \cite{choi2024joint}. 
We note that this approach is also aligned with the key principle of generalized mutual information \cite{yoo2006capacity,medard2000effect,lapidoth2002fading,ding2010maximum}.
Based on this approach, the lower bound for \eqref{eq:ergodic SE common} can be  derived  by treating $v_k$ as a Gaussian noise and using the Jensen's inequality as follows:
\begin{align}
    \label{eq:R_c,k lb}
    {\bar R}_{{\sf c},k}({\bf F})
    \geq &\; \mathbb{E}_{{\hat{\bf h}}_k}{\left[
    \log_2 \left(1 + \frac{|\hat{{\bf h}}^{\sf H}_k{\bf{\Phi}}_{{\alpha }_{\sf DAC}}{\bf{f}}_{0}|^2}
    {
    \text{IUI}_{\text c}'({\bf F})
    + {\rm{QE}}_k'({\bf F})
    + \frac{{\sigma}^2}{P}}
    \right) \right]} 
    \\\nonumber
    = &\; \mathbb{E}_{{\hat{\bf h}}_k}{\left[{{R}}_{{\sf c},k}^{{\sf {lb}}}({\bf F};\hat{\bf h}_k, {\bf{R}}_{\Tilde{{\bf h}}_k})\right]} 
    \\\label{eq:instanteous com rate}
    = & {\bar{R}}_{{\sf c},k}^{{\sf {lb}}}({\bf F};\hat{\bf h}_k, {\bf{R}}_{\Tilde{{\bf h}}_k}),
\end{align}
where 
\begin{align}
    \text{IUI}_{\text c}'({\bf F}) 
    = &\; {\sum_{i = 1}^{K}}{\bf{f}}_{i}^{\sf H}{\bf{\Phi}}_{{\alpha }_{\sf DAC}}^{\sf H} \left( {\hat{\bf h}}_k{\hat{\bf h}}^{\sf H}_k + {\bf{R}}_{\Tilde{{\bf h}}_k} \right) {\bf{\Phi}}_{{\alpha }_{\sf DAC}} {\bf{f}}_{i} \\
    \nonumber
    &+ {\bf{f}}_{0}^{\sf H}{\bf{\Phi}}_{{\alpha }_{\sf DAC}}^{\sf H}{\bf{R}}_{\Tilde{{\bf h}}_k}  {\bf{\Phi}}_{{\alpha }_{\sf DAC}} {\bf{f}}_{0}, \\
    \text{QE}_{k}'({\bf F})
    = &\; {\sum_{i = 0}^{K}} {\bf{f}}_{i}^{\sf H} {\bf{\Phi}}_{{\alpha }_{\sf DAC}}{\bf{\Phi}}_{{\beta }_{\sf DAC}} {\rm diag} \left( {\hat{\bf h}}_k{\hat{\bf h}_k}^{\sf H} + {\bf{R}}_{\Tilde{{\bf h}}_k} \right) {\bf{f}}_{i}.
\end{align}
A lower bound for ${\bar R_{k}}({\bf F})$, denoted as ${\bar R}_{k}^{{\sf {lb}}}({\bf F};\hat{\bf h}_k, {\bf{R}}_{\Tilde{{\bf h}}_k})$, can be similarly derived as:
\begin{align}
    {\bar R}_{k}({\bf F})
    \label{eq:R_k lb}
    \geq &\; \mathbb{E}_{{\hat{\bf h}}_k}{\left[
    \log_2 \left(1 + \frac{|\hat{{\bf h}}^{\sf H}_k{\bf{\Phi}}_{{\alpha }_{\sf DAC}}{\bf{f}}_{k}|^2}
    {
    \text{IUI}_{k}'({\bf F})
    + {\rm{QE}}_k'({\bf F})
    + \frac{{\sigma}^2}{P}}
    \right) \right]} 
    \\\nonumber
    = &\; \mathbb{E}_{{\hat{\bf h}}_k}{\left[{{R}}_{k}^{{\sf {lb}}}({\bf F};\hat{\bf h}_k, {\bf{R}}_{\Tilde{{\bf h}}_k})\right]} 
    \\\label{eq:instanteous priv rate}
    =& {\bar{R}}_{k}^{{\sf {lb}}}({\bf F};\hat{\bf h}_k, {\bf{R}}_{\Tilde{{\bf h}}_k}),
\end{align}
where 
\begin{align}
    \text{IUI}_{k}'({\bf F})
    \nonumber
    = &\; {\sum_{i = 1, i \neq k}^{K}} {\bf{f}}_{i}^{\sf H}{\bf{\Phi}}_{{\alpha }_{\sf DAC}}^{\sf H} \left( {\hat{\bf h}}_k{\hat{\bf h}}^{\sf H}_k + {\bf{R}}_{\Tilde{{\bf h}}_k} \right) {\bf{\Phi}}_{{\alpha }_{\sf DAC}} {\bf{f}}_{i} \\
    &  + {\bf{f}}_{k}^{\sf H}{\bf{\Phi}}_{{\alpha }_{\sf DAC}}^{\sf H}{\bf{R}}_{\Tilde{{\bf h}}_k}  {\bf{\Phi}}_{{\alpha }_{\sf DAC}} {\bf{f}}_{k}.
\end{align}

Using the expressions in \eqref{eq:R_c,k lb} and \eqref{eq:R_k lb}, we can now express the ergodic sum SE as:
\begin{align}   
    \nonumber
    {\bar R}_{\Sigma}(\bF) &= \min_{k} \! \left\{\mathbb{E}_{{\hat{\bf h}}_k} \!\!
    \left[ \! { R}_{{\sf c},k}^{{\sf {lb}}}({\bf F};\hat{\bf h}_k, {\bf{R}}_{\Tilde{{\bf h}}_k}) \! \right] \! \right\} \!
    +\! \sum_{k=1}^{K}\mathbb{E}_{{\hat{\bf h}}_k} \!\!
    \left[\!{{R}}_{k}^{{\sf {lb}}}({\bf F};\hat{\bf h}_k, {\bf{R}}_{\Tilde{{\bf h}}_k})\!\right] 
    \\\nonumber
    &\ge \mathbb{E}_{{\hat{\bf h}}_k} {\left[\min_{k}\left\{{ R}_{{\sf c},k}^{{\sf {lb}}}({\bf F};\hat{\bf h}_k, {\bf{R}}_{\Tilde{{\bf h}}_k})\right\} + \sum_{k=1}^{K}{{R}}_{k}^{{\sf {lb}}}({\bf F};\hat{\bf h}_k, {\bf{R}}_{\Tilde{{\bf h}}_k})\right]} 
    \\
    \label{eq:ergodic sum SE}
    &= \mathbb{E}_{{\hat{\bf h}}_k} {\left[{ R}_\Sigma^{\sf lb}({\bf F};\hat{\bf h}_k, {\bf{R}}_{\Tilde{{\bf h}}_k})\right]},
\end{align}
where
\begin{align}
    \label{eq:Rsum_lb}
   {R}_\Sigma^{\sf lb}({\bf F};\hat{\bf h}_k, {\bf{R}}_{\Tilde{{\bf h}}_k}) 
   = \min_{k} \! \left\{{ R}_{{\sf c},k}^{{\sf {lb}}}({\bf F};\hat{\bf h}_k, {\bf{R}}_{\Tilde{{\bf h}}_k})\right\} \!
   + \! \sum_{k=1}^{K} \! {{R}}_{k}^{{\sf {lb}}}({\bf F};\hat{\bf h}_k, {\bf{R}}_{\Tilde{{\bf h}}_k}).
\end{align}
It is important to note that since the BS has knowledge of the estimated channel $\hat{\bf h}_k$, we utilize the derived lower bound with the conditional average over the channel estimation error ${R}_\Sigma^{\sf lb}({\bf F};\hat{\bf h}_k, {\bf{R}}_{\Tilde{{\bf h}}_k})$ in \eqref{eq:Rsum_lb}, in place of the instantaneous SE ${R}_\Sigma(\bF)$ in the original problem in \eqref{eq:problem_main}. 
This approach enables us to leverage the information provided by both the estimated channel and its corresponding error covariance matrix. 
Consequently, we reformulate the optimization problem in \eqref{eq:problem_main} by replacing $R_\Sigma({\bf F}(\mathcal{A}))$ with ${R}_\Sigma^{\sf lb}({\bf F}(\mathcal{A});\hat{\bf h}_k, {\bf{R}}_{\Tilde{{\bf h}}_k})$:
\begin{subequations} 
    \label{eq:problem_main_lb}
    \begin{align}
        \mathop{{\text{maximize}}}_{\mathcal{A},{\bf{F}}} & \;\;
        {R}_\Sigma^{\sf lb}({\bf F}(\mathcal{A});\hat{\bf h}_k, {\bf{R}}_{\Tilde{{\bf h}}_k})
        \\
        \label{eq:problem_main_constraint_lb}
        \mathop{{\text{subject to}}}& \;\; \varsigma^{-1} P_{\sf tx}({\bf F}(\mathcal{A})) + P_{\sf cir}(\mathcal{A}) \leq P_{\sf tot}, 
        \\
        & \;\; P_{\sf tx}({\bf F}(\mathcal{A})) \leq P.
    \end{align}
\end{subequations}

\subsection{GPI-Friendly Problem Reformulation}

To employ the GPI method \cite{choi2019GPI}, we first divide the problem in \eqref{eq:problem_main_lb} into two subproblems by decomposing the precoder ${\bf F}$ into a normalized precoding matrix ${\bf W}$ with fixed power that determines the precoding direction, and a power scaling parameter ${\tau}$ that determines the transmit power $P_{\sf tx}$.

Let  $P_{\sf tx}$ be expressed in terms of the scalar weight $\tau \in (0, 1]$ and the maximum transmit power $P$ as
\begin{align}
    \label{eq:P_tx}
    P_{\sf tx}(\tau) = {\text {Tr}}\left( \mathbb{E}{\left[{\bf x}_{\sf q}{\bf x}_{\sf q}^{\sf H}\right]}\right) = {\tau}P.
\end{align}
Since ${\text {Tr}}( \mathbb{E}{[ {\bf x}_{\sf q}{\bf x}_{\sf q}^{\sf H} ]} ) 
= {\text {Tr}}(P{\bf{\Phi}}_{{\alpha }_{\sf DAC}}{\bf F \bf F^{\sf H}})$ from the definition of the trace operator and ${\bf{\Phi}}_{{\alpha }_{\sf DAC}} = {\bf I}_N-{\bf{\Phi}}_{{\beta}_{\sf DAC}}$,  \eqref{eq:P_tx} can be rewritten as
\begin{align}
    \label{eq:power scaling}
    {\text {Tr}}\left({\bf{\Phi}}_{{\alpha }_{\sf DAC}}{\bf F \bf F^{\sf H}}\right) = {\tau}.
\end{align}
To incorporate the transmit power scaling parameter $\tau$ into the precoder, we define a normalized precoding matrix ${\bf W} = \left[{\bf w}_0,{\bf w}_1, \cdots, {\bf w}_K\right]$, where each precoder ${\bf w}_k$ is expressed as
\begin{align}
    \label{eq:weighted precoder}
    {\bf w}_k = \frac{1}{\sqrt{\tau}} {\bf{\Phi}}_{{\alpha }_{\sf DAC}}^{1/2}{\bf f}_k.
\end{align}
Then, by using ${\bf W}$, \eqref{eq:power scaling} can be rewritten as
\begin{align}
    \label{eq:trace_W}
    {\text {Tr}}\left({\bf W \bf W^{\sf H}}\right) = 1.
\end{align}
By using this normalized precoder, the problem in \eqref{eq:problem_main_lb} can be decomposed into two parts: the optimization of the precoding direction component ${\bf W}$ and the transmit power parameter $\tau$.

We now turn our attention to optimizing the normalized precoder ${\bf W}$. To find ${\bf W}$ for a fixed $\tau$, we use the GPI algorithm \cite{choi2019GPI} and subsequently reformulate the problem in \eqref{eq:problem_main_lb} into a GPI-friendly form.
Let ${\bar{\bf w}} = {{\rm vec}}\left(\bf W\right)$ be the vectorized normalized precoder, which can be seen as a precoding direction as $\|\bar {\bf w}\|^2=1$ from \eqref{eq:trace_W}.
By replacing $\frac{\sigma^2}{P}$ with $\frac{\sigma^2}{P}\|\bar {\bf w}\|^2$, we represent ${R}_{{\sf c},k}^{{\sf {lb}}}({\bf F};\hat{\bf h}_k, {\bf{R}}_{\Tilde{{\bf h}}_k})$ in \eqref{eq:R_c,k lb} into a Rayleigh quotient form  in logarithm as a function of $\bar{\bf w}$:
\begin{align}
    \label{eq:rewrite_block}
    {R}_{{\sf c},k}^{{\sf {lb}}}({\bar{\bf w}}, \tau) = \log_2 \left( \frac{\bar {\bf{w}}^{\sf H} {\bf{A}}_{{\sf c},k} \bar{\bf{w}}}{\bar{\bf{w}}^{\sf H} {\bf{B}}_{{\sf c},k} \bar{\bf{w}} } \right),
\end{align}
where
\begin{align}
    \nonumber
    & {\bf G}_k = \left({\bf{\Phi}}_{{\alpha }_{\sf DAC}}^{1/2}\right)^{\sf H} \left( {\hat{\bf h}}_k{\hat{\bf h}}^{\sf H}_k + {\bf{R}}_{\Tilde{{\bf h}}_k} \right) {\bf{\Phi}}_{{\alpha }_{\sf DAC}}^{1/2} \\
    &\;\;\;\;\;\;\;\;\;\;\;\;\;\; 
    + {\bf{\Phi}}_{{\beta }_{\sf DAC}}{\rm diag}\left( {\hat{\bf h}}_k{\hat{\bf h}}^{\sf H}_k + {\bf{R}}_{\Tilde{{\bf h}}_k} \right), \\
    & {\bf{A}}_{{\sf c},k} \! = {\rm blkdiag} \left({\bf G}_k, \cdots , {\bf G}_k \right) + {\bf{I}}_{N(K+1)} \frac{ \sigma^2}{\tau P}, \\
    & {\bf{B}}_{{\sf c},k} \!
    = \! {\bf{A}}_{{\sf c},k} \!
    -\! {\rm blkdiag} \! \left( \! \left( {{\bf{\Phi}}_{{\alpha }_{\sf DAC}}^{1/2}}\right)^{\! \sf H} \! \hat{{\bf h}}_k \hat{{\bf h}}^{\sf H}_k {\bf{\Phi}}_{{\alpha }_{\sf DAC}}^{1/2},{{\bf{0}}_{N}},\!\cdots\!, {{\bf{0}}_N}\!\right)\!.   
\end{align}
Here, ${\bf{A}}_{{\sf c},k}$ and ${\bf{B}}_{{\sf c},k}$ both represent block diagonal matrices of size $N(K+1)\times N(K+1)$.
Similarly, we also represent $R_{k}^{\sf lb}({\bf F};\hat{\bf h}_k, {\bf{R}}_{\Tilde{{\bf h}}_k})$ in \eqref{eq:R_k lb} as a Rayleigh quotient form:
\begin{align}
    \label{eq:rewrite_block_pri}
    R_{k}^{\sf lb}({\bar{\bf w}}, \tau) = \log_2 \left( \frac{\bar {\bf{w}}^{\sf H} {\bf{A}}_{k}\bar{\bf{w}}}{\bar{\bf{w}}^{\sf H} {\bf{B}}_{k} \bar{\bf{w}} } \right),
\end{align}
where
\begin{align}
    & {\bf G}_{{\sf c}, k} 
    = \left({\bf{\Phi}}_{{\alpha }_{\sf DAC}}^{1/2}\right)^{\sf H} \left( {\hat{\bf h}}_k{\hat{\bf h}}^{\sf H}_k + {\bf{R}}_{\Tilde{{\bf h}}_k} \right) {\bf{\Phi}}_{{\alpha }_{\sf DAC}}^{1/2}, \\
    & {\bf{A}}_{k}\! = {\rm blkdiag} \left( {\bf G}_k\!
    - {\bf G}_{{\sf c}, k}\!
    ,{\bf G}_k,\! \cdots \!, {\bf G}_k\right)\!
    + {\bf{I}}_{N(K+1)} \frac{ \sigma^2}{\tau P}, 
    \\
    & {\bf{B}}_{k} \! = 
    {\bf{A}}_{k} \!  - \! {\rm blkdiag} \! \bigg( \! {{\bf{0}}_N}\!,\! \cdots\! , \! \underbrace{ \left(\!{{\bf{\Phi}}_{{\alpha }_{\sf DAC}}^{1/2}}\! \right)^{\! \! \sf H} \! \! \hat{{\bf h}}_k \hat{{\bf h}}^{\sf H}_k {{\bf{\Phi}}_{{\alpha }_{\sf DAC}}^{1/2}}}_{{\text{the} \;(k+1){\text{th block}}}}, \! {{\bf{0}}_N} \! , \! \cdots\! , \! {{\bf{0}}_N} \! \! \bigg) \! .
\end{align}
Incorporating $R_{{\sf c},k}^{\sf lb}({\bar{\bf w}}, \tau)$ in \eqref{eq:rewrite_block} and $R_{k}^{\sf lb}({\bar{\bf w}}, \tau)$ in \eqref{eq:rewrite_block_pri} into the formula for ${R}_\Sigma^{\sf lb}({\bf F};\hat{\bf h}_k, {\bf{R}}_{\Tilde{{\bf h}}_k})$ in \eqref{eq:Rsum_lb}, we have
\begin{align}
    \label{eq:Rsum_lb (W) min}
   {R}_\Sigma^{\sf lb}({\bar{\bf w}}, \tau) 
   = \min_{k} \! \left\{
   \log_2 \! \left( \frac{\bar {\bf{w}}^{\sf H} {\bf{A}}_{{\sf c},k} \bar{\bf{w}}}{\bar{\bf{w}}^{\sf H} {\bf{B}}_{{\sf c},k} \bar{\bf{w}} } \right) \!
   \right\}
   \! + \! \sum_{k=1}^{K} \!
   \log_2 \! \left( \frac{\bar {\bf{w}}^{\sf H} {\bf{A}}_{k}\bar{\bf{w}}}{\bar{\bf{w}}^{\sf H} {\bf{B}}_{k} \bar{\bf{w}} } \right) \!.
\end{align}
In \eqref{eq:Rsum_lb (W) min}, there remains a non-smoothness issue of the $\min\{\cdot\}$ function.
This obstacle can be avoided by using the LogSumExp technique \cite{shen2010dual} to approximate the term with the $\min\{\cdot\}$ function such that it is in a smooth form:
\begin{align}
    \label{eq:approxcom}
    \nonumber
    \min_{k \in \CMcal{K }} \{R_{{\sf c},k}^{\sf lb}({\bar{\bf w}}, \tau)\} 
    \approx &\; -{a} \ln \left( \sum_{k = 1}^{K} \exp\left( -\frac{1}{a} R_{{\sf c},k}^{\sf lb}({\bar{\bf w}}, \tau) \right) \right) \\
    = &\;  -a \ln \left(\sum_{k = 1}^{K} \left( \frac{\bar {\bf{w}}^{\sf H} {\bf{A}}_{{\sf c},k} \bar {\bf{w}}}{\bar {\bf{w}}^{\sf H} {\bf{B}}_{{\sf c},k}  \bar {\bf{w}} }  \right)^{-\frac{1}{a\ln2}} \right).
\end{align}
Here, the approximation becomes tight as the positive constant $a$ approaches $0^+$.

Meanwhile, we also approximate the indicator function $\mathbbm{1}_{\left\{i \in \mathcal{A} \right\}}$ in \eqref{antenna power} using the following approximation \cite{sriperumbudur2011majorization}:
\begin{align}
    \label{rho approximation}
    \mathbbm{1}_{\left\{ |x|^2 > 0 \right\}} \approx \frac{\log_2(1 + |x|^2/\rho)}{\log_2(1 + 1/\rho)},
\end{align}
where the approximation becomes tight as the positive constant $\rho$ approaches $0^+$.
We note that the approximation in \eqref{rho approximation} can be verified by using L'H\^opital's rule \cite{sriperumbudur2011majorization}.
If we denote ${\Tilde{\bf f}}_i$ as the $i$th row of the precoder ${\bf F}$, we remark that the $i$th antenna is active if and only if $\lVert {\Tilde{\bf f}}_i \rVert^2 > 0$.
Since ${\bf F} \propto {{\bf{\Phi}}_{{\alpha }_{\sf DAC}}^{-1/2}} \bf W$, 
$P_{\sf ant}(\mathcal{A})$ in \eqref{antenna power} can be approximated as
\begin{align}
    \label{antenna power2}
    \nonumber
    P_{\sf ant}({\bf F})
    = &\; {\sum_{i=1}^{N}}\mathbbm{1}_{\left\{\lVert {\Tilde{\bf f}}_i \rVert^2 > 0\right\}}P_{{\sf ant}, i}
    = {\sum_{i=1}^{N}}\mathbbm{1}_{\left\{\lVert {\Tilde{\bf w}}_i / \sqrt{\alpha_{{\sf DAC}, i}} \rVert^2 > 0\right\}}P_{{\sf ant}, i} \\
    \approx &\; {\sum_{i=1}^{N}}\log_2 \left(1+\rho^{-1}\left\lVert \frac{{\Tilde{\bf w}}_i}{\sqrt{\alpha_{{\sf DAC}, i}}} \right\rVert^2 \right)^\frac{P_{{\sf ant}, i}}{\log_2(1+\rho^{-1})}.
\end{align}
This approximation not only makes the power constraint in \eqref{eq:problem_main_constraint_lb} more tractable but also allows a joint optimization of precoding and antenna selection without explicitly considering $\mathcal{A}$.
We remark that the accuracy of the approximation in \eqref{antenna power2} is sufficient to distinguish between active and inactive antennas.
Let ${\bf e}_i$ denote the $N$-dimensional standard basis vector with 1 at the $i$th element and zeros elsewhere.
Denoting the $i$th row of ${\bf W}$ as ${\Tilde{\bf w}}_i$,
we have ${\Tilde{\bf w}}_i / \sqrt{\alpha_{{\sf DAC}, i}} = {\bf e}_i^{\sf H} {{\bf{\Phi}}_{{\alpha }_{\sf DAC}}^{-1/2}} {\bf W}$, and thus, the norm in \eqref{antenna power2} can be computed as
\begin{align}
    \nonumber
    \left\lVert \frac{{\Tilde{\bf w}}_i}{\sqrt{\alpha_{{\sf DAC}, i}}} \right\rVert^2
    = &\; {\bf e}_i^{\sf H} {{\bf{\Phi}}_{{\alpha }_{\sf DAC}}^{-1/2}} {\bf W}{\bf W}^{\sf H} {{\bf{\Phi}}_{{\alpha }_{\sf DAC}}^{-1/2}} {\bf e}_i \\
    \nonumber
    = &\; \text{vec}({\Tilde{\bf e}}_i^{\sf H} {\bf W}{\bf W}^{\sf H} {\Tilde{\bf e}}_i) \\
    \nonumber
    \stackrel{\mathclap{(a)}}{=} &\; \left(\left({\Tilde{\bf e}}_i^{\sf T}{\bf W}^*\right) \otimes \Tilde{\bf e}_i^{\sf H}\right)\text{vec}\left( {\bf W} \right) \\
    \nonumber
    \stackrel{\mathclap{(b)}}{=} &\; \left[ \left\{ \left( {\bf I}_{K+1} \otimes \Tilde{\bf e}_i^{\sf T} \right) \text{vec}\left({\bf W}^*)\right) \right\}^{\sf T} \otimes \Tilde{\bf e}_i^{\sf H} \right] \text{vec}\left({\bf W}\right) \\
    \stackrel{\mathclap{(c)}}{=} &\; {\bar{\bf w}}^{\sf H} \left( {\bf I}_{K+1} \otimes \Tilde{\bf e}_i \otimes \Tilde{\bf e}_i^{\sf H} \right) {\bar{\bf w}},
\end{align}
where $\Tilde{\bf e}_i = {{\bf{\Phi}}_{{\alpha }_{\sf DAC}}^{-1/2}} {\bf e}_i$, steps (a) and (b) are based on $\text{vec}(\bf ABC) = (\bf{C}^{\sf T} \otimes \bf{A} ) \text{vec}(\bf B)$, and step (c) is based on $({\bf A} \otimes {\bf B})^{\sf T} = {\bf A}^{\sf T} \otimes {\bf B}^{\sf T}$.
Consequently, the approximation of $P_{\sf ant}({\bf F})$ in \eqref{antenna power2} can be rewritten as
\begin{align}
    \nonumber
    \Tilde{P}_{\sf ant}({\bar{\bf w}})
    = &\; {\sum_{i=1}^{N}}\log_2 \left(1+\rho^{-1} {\bar{\bf w}}^{\sf H} {\bf I}_{K+1} \otimes \Tilde{\bf e}_i \otimes \Tilde{\bf e}_i^{\sf H} {\bar{\bf w}} \right)^\frac{P_{{\sf ant}, i}}{\log_2(1+\rho^{-1})} 
    \\ \label{eq:antenna power3}
    = &\; \frac{1}{\ln{2}} {\sum_{i=1}^{N}}\ln \left( {\bar{\bf w}}^{\sf H} {\bf E}_i {\bar{\bf w}} \right)^\frac{P_{{\sf ant}, i}}{\log_2(1+\rho^{-1})},
\end{align}
where ${\bf E}_i = {\bf I}_{N(K+1)} + \rho^{-1}{\bf I}_{K+1} \otimes \Tilde{\bf e}_i \otimes \Tilde{\bf e}_i^{\sf H}$.
Using \eqref{eq:antenna power3}, $P_{\sf cir}(\mathcal{A})$ in \eqref{circuit power} can also be rewritten as $\Tilde{P}_{\sf cir}({\bar{\bf w}}) = P_{\sf LO} + \Tilde{P}_{\sf ant}({\bar{\bf w}})$.
Finally, based on \eqref{eq:rewrite_block_pri}, \eqref{eq:approxcom}, and \eqref{eq:antenna power3}, we reformulate the optimization problem in \eqref{eq:problem_main_lb}  without explicitly considering
$\mathcal{A}$ as follows:
\begin{subequations} \label{eq:problem_new_Blk_q}
    \begin{align}
        \mathop{{\text{maximize}}}_{\bar {\bf{w}},\tau}& \;\;
        \ln \left( \! \sum_{k = 1}^{K} \! \left( \! \frac{\bar {\bf{w}}^{\sf H} {\bf{A}}_{{\sf c},k} \bar {\bf{w}}}{\bar {\bf{w}}^{\sf H} {\bf{B}}_{{\sf c},k}  \bar {\bf{w}} } \! \right)^{\! \! -\frac{1}{a\ln \! 2}} \! \right)^{\! \!-a}
        \! \! \! + \!
        \sum_{k = 1}^{K} \! \ln \! \left( \! \frac{\bar {\bf{w}}^{\sf H} {\bf{A}}_k \bar {\bf{w}}}{\bar {\bf{w}}^{\sf H} {\bf{B}}_k \bar {\bf{w}}} \! \right)^{\! \! \frac{1}{\ln \! 2}} \\
        \label{eq:constraint_new_new_Blk}
        {\text{subject to}} & \;\; 
        \frac{\tau}{\varsigma} P
        \! + \! 
        P_{\sf LO} 
        \! + \! \! {\sum_{i=1}^{N}} \! \ln \! \left( \! {\bar{\bf w}}^{\sf H} {\bf E}_i {\bar{\bf w}} \! \right)^{\! \frac{P_{{\sf ant}, i}}{(\ln{ \!2}) \! \log_2 \! (1+\rho^{-1})}}
        \! \leq P_{\sf tot},
        \\\label{eq:const_bar_w}
        &\;\; \|{\bar{{\bf w}}}\| = 1, 
        \\\label{eq:const_tau}
        &\;\; \tau \in (0, 1].
    \end{align}
\end{subequations}
We note that even without explicit consideration of $\mathcal{A}$, \eqref{eq:constraint_new_new_Blk} can enforce the algorithm to selectively turn off certain antennas to reduce the circuit power.

To solve the problem in \eqref{eq:problem_new_Blk_q}, we  temporarily relax the constraints \eqref{eq:const_bar_w} and \eqref{eq:const_tau} and reconsider them when proposing the algorithms.
Then, ignoring the constraints \eqref{eq:const_bar_w} and \eqref{eq:const_tau}, we  express the Lagrangian function of 
the problem in \eqref{eq:problem_new_Blk_q} as
\begin{align}
    \nonumber
    L({\bar{{\bf w}}}, \tau, \mu) 
    = &\; 
    \ln \left( \sum_{k = 1}^{K} \! \left(  \frac{\bar {\bf{w}}^{\sf H} {\bf{A}}_{{\sf c},k} \bar {\bf{w}}}{\bar {\bf{w}}^{\sf H} {\bf{B}}_{{\sf c},k}  \bar {\bf{w}} }  \right)^{\! - \frac{1}{a\ln2}}  \right)^{\! -a}
    \! \! + \sum_{k = 1}^{K}  \ln \! \left(  \frac{\bar {\bf{w}}^{\sf H} {\bf{A}}_k \bar {\bf{w}}}{\bar {\bf{w}}^{\sf H} {\bf{B}}_k \bar {\bf{w}}}  \right)^{\! \frac{1}{\ln2}} \\
    \nonumber
    & \quad
    - \mu \left(
        \frac{\tau}{\varsigma} P
        \! + \! 
        P_{\sf LO} 
        \! + \! \! {\sum_{i=1}^{N}} \! \ln \! \left( \! {\bar{\bf w}}^{\sf H} {\bf E}_i {\bar{\bf w}} \! \right)^{\! \frac{P_{{\sf ant}, i}}{(\ln{ \!2}) \! \log_2 \! (1+\rho^{-1})}}
        \! \! - \! P_{\sf tot}
    \right)
    \\
    \label{lagrangian}
    =& \ln{\lambda({\bar{{\bf w}}}, \tau, \mu)},
\end{align}
where $\mu$ is the Lagrange multiplier, and $\lambda({\bar{{\bf w}}}, \tau, \mu)$ is
\begin{align}
    &\; \lambda({\bar{{\bf w}}}, \tau, \mu) = \\
    \nonumber
    &\; \left\{ \sum_{k=1}^{K} \! \left( \frac{\bar {\bf{w}}^{\sf H} {\bf{A}}_{{\sf c},k} \bar {\bf{w}}}{\bar {\bf{w}}^{\sf H} {\bf{B}}_{{\sf c},k}  \bar {\bf{w}} }  \right)^{\!\! -\frac{1}{a\ln2}} \right\}^{\!\!-a}
    \frac{e^{-\mu \left( \frac{\tau}{\varsigma}P - P_{\sf tot} + P_{\sf LO} \right) } \prod\limits_{l=1}^{K} \left( \frac{\bar {\bf{w}}^{\sf H} {\bf{A}}_{l} \bar {\bf{w}}}{\bar {\bf{w}}^{\sf H} {\bf{B}}_{l}  \bar {\bf{w}} }  \right)^{\frac{1}{\ln2}} }{\prod\limits_{i=1}^{N} \left( \bar {\bf{w}}^{\sf H} {\bf{E}}_{i} \bar {\bf{w}} \right)^ \frac{\mu P_{{\sf ant}, i}}{(\ln{2}) \log_2(1+\rho^{-1})} }.
\end{align}

\subsection{Precoding Direction Optimization}

To optimize $\bar {\bf w}$ in \eqref{lagrangian}, we first derive the Karush-Kuhn-Tucker (KKT) stationary condition 
and cast the condition into a generalized eigenvalue problem as shown in Lemma~\ref{lem:main}.
\begin{lemma}
    \label{lem:main}
    For given values of $\tau$ and $\mu$, the optimization problem in \eqref{lagrangian} satisfies the first-order optimality condition if the following criterion is met:
    \begin{align}
        \label{functional eigenvalue problem}
        {\bf B}_{\sf KKT}^{-1}({\bar{{\bf w}}}; \tau, \mu) {\bf A}_{\sf KKT}({\bar{{\bf w}}}; \tau, \mu) \bar{{\bf w}} = \lambda({\bar{{\bf w}}}; \tau, \mu)\bar{{\bf w}},
    \end{align}
    where
    \begin{align}
        \label{eq:lem_A_kkt_Q}
        & {\bf A}_{\sf KKT}({\bar{{\bf w}}}; \tau, \mu) = 
        \\\nonumber
        &\; \lambda_{\sf num}({\bar{{\bf w}}}; \tau, \mu)
        \sum\limits_{k = 1}^{K} \left\{
        \frac{
            \left( \frac{\bar {\bf{w}}^{\sf H} {\bf{A}}_{{\sf c},k} \bar{\bf{w}}}{\bar{\bf{w}}^{\sf H} {\bf{B}}_{{\sf c},k} \bar{\bf{w}} } \right)^{-\frac {1} {a\ln2}} \left( \frac{{\bf{A}}_{{\sf c},k}}{\bar{\bf{w}}^{\sf H} {\bf{A}}_{{\sf c},k} \bar{\bf{w}}} \right)
            }
            {
            \sum\limits_{l = 1}^{K} \left( \frac{\bar {\bf{w}}^{\sf H} {\bf{A}}_{{\sf c},l} \bar{\bf{w}}}{\bar{\bf{w}}^{\sf H} {\bf{B}}_{{\sf c},l} \bar{\bf{w}} } \right)^{-\frac {1} {a\ln2}}
            }
        + \frac{
            {\bf{A}}_k 
            }
            {\bar{\bf{w}}^{\sf H} {\bf{A}}_k \bar{\bf{w}}
            }
            \right\}, 
    \end{align}
    \begin{align}
        \nonumber
        & {\bf B}_{\sf KKT}({\bar{{\bf w}}}; \tau, \mu) =  \;  \lambda_{\sf{den}}({\bar{{\bf w}}}; \tau, \mu)\left(\frac{\mu}{\log_2(1+\rho^{-1})} \sum\limits_{i = 1}^{N}\frac{P_{{\sf ant}, i}{\bf{E}}_i }{\bar{\bf{w}}^{\sf H} {\bf{E}}_i \bar{\bf{w}} } \right.
        \\\label{eq:lem_B_kkt_Q}
        &\;  \left.
        \quad \quad +\sum\limits_{k = 1}^{K} 
        \frac{\left( \frac{\bar {\bf{w}}^{\sf H} {\bf{A}}_{{\sf c},k} \bar{\bf{w}}}{\bar{\bf{w}}^{\sf H} {\bf{B}}_{{\sf c},k} \bar{\bf{w}} } \right)^{-\frac {1} {a\ln2}} \left( \frac{{\bf{B}}_{{\sf c},k} }{\bar{\bf{w}}^{\sf H} {\bf{B}}_{{\sf c},k} \bar{\bf{w}} } \right) }{\sum\limits_{l = 1}^{K} \left( \frac{\bar {\bf{w}}^{\sf H} {\bf{A}}_{{\sf c},l} \bar{\bf{w}}}{\bar{\bf{w}}^{\sf H} {\bf{B}}_{{\sf c},l} \bar{\bf{w}} } \right)^{-\frac {1} {a\ln2}}}
        + \sum\limits_{l = 1}^{K}\frac{{\bf{B}}_k }{\bar{\bf{w}}^{\sf H} {\bf{B}}_k \bar{\bf{w}} } \right).
    \end{align}
    Here, $\lambda_{\sf num}({\bar{{\bf w}}}; \tau, \mu)$ and $\lambda_{\sf den}({\bar{{\bf w}}}; \tau, \mu)$ can be any pair of values that satisfies $\lambda({\bar{{\bf w}}}; \tau, \mu) = \frac{\lambda_{\sf num}({\bar{{\bf w}}}; \tau, \mu)}{\lambda_{\sf den}({\bar{{\bf w}}}; \tau, \mu)}$.
\end{lemma}
\begin{proof}
    We find the stationarity condition by setting the partial derivative of the objective function in \eqref{lagrangian} to zero.
    For convenience, let $L_1({\bar{{\bf w}}}; \tau)$, $L_2({\bar{{\bf w}}}; \tau)$, and $L_3({\bar{{\bf w}}}; \tau, \mu)$ be
    \begin{align}
        \label{L1}
        & L_1({\bar{{\bf w}}}; \tau) 
        = -a \ln \left( \sum_{k = 1}^{K} \left( \frac{\bar {\bf{w}}^{\sf H} {\bf{A}}_{{\sf c},k} \bar{\bf{w}}}{\bar{\bf{w}}^{\sf H} {\bf{B}}_{{\sf c},k} \bar{\bf{w}} } \right)^{-\frac {1} {a\ln2}} \right), \\
        \label{L2}
        & L_2({\bar{{\bf w}}}; \tau) 
        = \frac{1}{\ln{2}}{\sum_{k=1}^{K}}\ln \left( \frac{\bar {\bf{w}}^{\sf H} {\bf{A}}_{k}\bar{\bf{w}}}{\bar{\bf{w}}^{\sf H} {\bf{B}}_{k} \bar{\bf{w}} } \right), \\
        \label{L3}
        & L_3({\bar{{\bf w}}}; \tau, \mu) = 
        - \mu \left(
            \frac{\tau}{\varsigma} P
            \! + \! 
            P_{\sf LO} 
            \! + \! \! 
            {\sum_{i=1}^{N}} \! \ln \! \left( \! {\bar{\bf w}}^{\sf H} {\bf E}_i {\bar{\bf w}} \! \right)^{\! \frac{P_{{\sf ant}, i}}{(\ln{ \!2}) \! \log_2 \! (1+\rho^{-1})}}
            \! \! - \! P_{\sf tot}
        \right),
    \end{align}
    respectively. Then, the partial derivatives of $L_1({\bar{{\bf w}}}; \tau)$, $L_2({\bar{{\bf w}}}; \tau)$, and $L_3({\bar{{\bf w}}}; \tau, \mu)$ can be derived as
    \begin{align}
        & \frac{\partial L_1({\bar{{\bf w}}}; \tau)}{\partial {\bar{\bf w}}^{\sf H}} 
        = \frac{ \sum\limits_{k = 1}^{K} \left( \frac{\bar {\bf{w}}^{\sf H} {\bf{A}}_{{\sf c},k} \bar{\bf{w}}}{\bar{\bf{w}}^{\sf H} {\bf{B}}_{{\sf c},k} \bar{\bf{w}} } \right)^{-\frac {1} {a\ln2}} \left( \frac{{\bf{A}}_{{\sf c},k} \bar{\bf{w}}}{\bar{\bf{w}}^{\sf H} {\bf{A}}_{{\sf c},k} \bar{\bf{w}} } - \frac{{\bf{B}}_{{\sf c},k} \bar{\bf{w}}}{\bar{\bf{w}}^{\sf H} {\bf{B}}_{{\sf c},k} \bar{\bf{w}} } \right) }{(\ln{2}) \sum\limits_{l = 1}^{K} \left( \frac{\bar {\bf{w}}^{\sf H} {\bf{A}}_{{\sf c},l} \bar{\bf{w}}}{\bar{\bf{w}}^{\sf H} {\bf{B}}_{{\sf c},l} \bar{\bf{w}} } \right)^{-\frac {1} {a\ln2}}}, \\
        & \frac{\partial L_2({\bar{{\bf w}}}; \tau)}{\partial {\bar{\bf w}}^{\sf H}} 
        = \frac{1}{\ln{2}} \sum\limits_{k = 1}^{K} \left( \frac{{\bf{A}}_k \bar{\bf{w}}}{\bar{\bf{w}}^{\sf H} {\bf{A}}_k \bar{\bf{w}} } - \frac{{\bf{B}}_k \bar{\bf{w}}}{\bar{\bf{w}}^{\sf H} {\bf{B}}_k \bar{\bf{w}} } \right), \\
        & \frac{\partial L_3({\bar{{\bf w}}}; \tau, \mu)}{\partial {\bar{\bf w}}^{\sf H}} 
        = \frac{-\mu}{(\ln{2}) \log_2(1+\rho^{-1})} \sum\limits_{i = 1}^{N}P_{{\sf ant}, i}\frac{{\bf{E}}_i \bar{\bf{w}}}{\bar{\bf{w}}^{\sf H} {\bf{E}}_i \bar{\bf{w}} }.
    \end{align}
    Subsequently, the first-order optimality condition is given as
    \begin{align}
        \label{kkt condition}
        \frac{\partial L({\bar{{\bf w}}}; \tau, \mu)}{\partial {\bar{\bf w}}^{\sf H}} 
        = \frac{\partial L_1({\bar{{\bf w}}}; \tau)}{\partial {\bar{\bf w}}^{\sf H}} 
        + \frac{\partial L_2({\bar{{\bf w}}}; \tau)}{\partial {\bar{\bf w}}^{\sf H}} 
        + \frac{\partial L_3({\bar{{\bf w}}}; \tau, \mu)}{\partial {\bar{\bf w}}^{\sf H}} = {\bf{0}}.
    \end{align}
    This can be rearranged into the form shown below:
    \begin{align}
        {\bf{A}}_{\sf KKT}(\bar {\bf{w}}; \tau, \mu) \bar {\bf{w}} = \lambda({\bar{{\bf w}}}; \tau, \mu) {\bf{B}}_{\sf KKT} (\bar {\bf{w}}; \tau, \mu) \bar {\bf{w}}.
    \end{align}
    Here, ${\bf{B}}_{\sf KKT}(\bar {\bf{w}}; \tau, \mu)$ is guaranteed to be invertible because ${\bf B}_{{\sf c},k}$ and ${\bf B}_{k}$ are Hermitian block diagonal matrices.
    As a result, we arrive at the condition in \eqref{functional eigenvalue problem}, which completes the proof.
\end{proof}

We regard the condition in \eqref{functional eigenvalue problem} as a functional eigenvalue problem.
In this regard, ${\bar{\bf w}}$ and $\lambda({\bar{{\bf w}}}; \tau, \mu)$ can be interpreted as an eigenvector and its corresponding eigenvalue of ${\bf{B}}_{\sf KKT}^{-1} (\bar {\bf{w}}; \tau, \mu){\bf{A}}_{\sf KKT}(\bar {\bf{w}}; \tau, \mu)$, respectively.
Then, finding the stationary solution of the problem can be viewed as a nonlinear eigenvalue problem with eigenvector dependency (NEPv) \cite{cai:siam:18}.
Noting the relationship of $\lambda(\bar {\bf w};\tau,\mu)$ to the Lagrangian function in \eqref{lagrangian}, finding the leading eigenvector in \eqref{functional eigenvalue problem} corresponds to identifying one of the stationary points that maximizes the Lagrangian function.
To find the principal eigenvector, we adopt an efficient GPI-based algorithm \cite{choi2019GPI}.

Let $\bar {\bf{w}}^{(t)}$ be the vectorized normalized precoder at iteration $t$.
Using \eqref{eq:lem_A_kkt_Q} and \eqref{eq:lem_B_kkt_Q}, we form the matrices  ${\bf{A}}_{\sf KKT} (\bar {\bf{w}}^{(t)}; \tau, \mu)$ and ${\bf{B}}_{\sf KKT} (\bar {\bf{w}}^{(t)}; \tau, \mu)$ and update $\bar {\bf{w}}^{(t+1)}$ using the following:
\begin{align}
    \label{eq:PowerIteration}
    \bar {\bf{w}}^{(t+1)} = \frac{
    {{\bf{B}}^{-1}_{\sf KKT} (\bar {\bf{w}}^{(t)}; \tau, \mu) {\bf{A}}_{\sf KKT} (\bar {\bf{w}}^{(t)}; \tau, \mu) \bar {\bf{w}}^{(t)}}
    }{
    \left\lVert {{\bf{B}}^{-1}_{\sf KKT} (\bar {\bf{w}}^{(t)}; \tau, \mu) {\bf{A}}_{\sf KKT} (\bar {\bf{w}}^{(t)}; \tau, \mu) \bar {\bf{w}}^{(t)}} \right\rVert
    }.
\end{align}
We remark that the normalization in \eqref{eq:PowerIteration} guarantees the feasibility of $\bar {\bf w}$, i.e., $\|\bar {\bf w} \|^2= 1$.
The algorithm continues this process until the convergence condition $\lVert {\bf{w}}^{(t+1)} - \bar {\bf{w}}^{(t)} \rVert \leq \epsilon$ is met, where $\epsilon$ is a positive threshold value, or until a predetermined maximum number of iterations $t_{\rm max}$ is reached. Algorithm~\ref{alg:main_QRS} outlines the steps of this procedure.

\begin{algorithm} [t]
\caption{Optimization of ${\bf W}$} \label{alg:main_QRS} 
{\bf{initialize}}: $\bar {\bf{w}}^{(0)}$\\
Set the iteration count $t = 0$\\
\While {$\left\lVert \bar{\bf{w}}^{(t+1)} - \bar{\bf{w}}^{(t)} \right\rVert > \epsilon$ $\it{\&}$ $t < t_{\rm max}$}{
Build matrix $ {\bf{A}}_{\sf KKT} (\bar {\bf{w}}^{(t)}; \tau, \mu)$ in \eqref{eq:lem_A_kkt_Q}\\
Build matrix $ {\bf{B}}_{\sf KKT} (\bar {\bf{w}}^{(t)}; \tau, \mu)$ in \eqref{eq:lem_B_kkt_Q} \\
Compute $\bar {\bf{w}}^{(t+1)} = $ ${{\bf{B}}^{-1}_{\sf KKT} (\bar {\bf{w}}^{(t)}; \tau, \mu) {\bf{A}}_{\sf KKT} (\bar {\bf{w}}^{(t)}; \tau, \mu) \bar {\bf{w}}^{(t)}}$ \\
Normalize $\bar {\bf{w}}^{(t+1)} \leftarrow  \bar {\bf{w}}^{(t+1)} /   \left\lVert \bar{\bf{w}}^{(t+1)} \right\rVert$\\
 $t \leftarrow t+1$}

\Return{\ }{$\bar{\bf w}^{(t)}$}
\end{algorithm}

\subsection{Joint  Precoding, Antenna Selection, and Power Control}

\begin{figure*}[b]
    \centering
    \rule{\textwidth}{0.7pt}
    \vspace{-1 em}
\begin{align}
    \label{eq:Xi_1}
    \Xi_1
    &= {\sum_{i = 1}^{K}} {\bf{w}}_{i}^{\sf H} \left({\bf{\Phi}}_{{\alpha }_{\sf DAC}}^{1/2}\right)^{\sf H} \left( {\hat{\bf h}}_k{\hat{\bf h}}^{\sf H}_k + {\bf{R}}_{\Tilde{{\bf h}}_k} \right) {\bf{\Phi}}_{{\alpha }_{\sf DAC}}^{1/2} {\bf{w}}_{i}
    + {\bf{w}}_{0}^{\sf H} \left({\bf{\Phi}}_{{\alpha }_{\sf DAC}}^{1/2}\right)^{\sf H} {\bf{R}}_{\Tilde{{\bf h}}_k} {\bf{\Phi}}_{{\alpha }_{\sf DAC}}^{1/2} {\bf{w}}_{0}
    + \sum\limits_{i=0}^{K}{\bf w}_i^{\sf H}{\bf{\Phi}}_{{\beta }_{\sf DAC}}{\rm diag}\left( {\hat{\bf h}}_k{\hat{\bf h}_k}^{\sf H} + {\bf{R}}_{\Tilde{{\bf h}}_k} \right){\bf w}_i
    + \frac{{\sigma}^2}{\tau P}, \\
    \label{eq:Xi_2}
    \Xi_2
    &= \! \! \sum_{i = 1, i \neq k}^K \! \!
    {\bf{w}}_{i}^{\sf H} \left({\bf{\Phi}}_{{\alpha }_{\sf DAC}}^{1/2}\right)^{\! \sf H} \left( {\hat{\bf h}}_k{\hat{\bf h}}^{\sf H}_k + {\bf{R}}_{\Tilde{{\bf h}}_k} \right) {\bf{\Phi}}_{{\alpha }_{\sf DAC}}^{1/2} {\bf{w}}_{i}
    + {\bf{w}}_{k}^{\sf H} \left({\bf{\Phi}}_{{\alpha }_{\sf DAC}}^{1/2}\right)^{\! \sf H} {\bf{R}}_{\Tilde{{\bf h}}_k} {\bf{\Phi}}_{{\alpha }_{\sf DAC}}^{1/2} {\bf{w}}_{k}
    + \sum\limits_{i=0}^{K}{\bf w}_i^{\sf H}{\bf{\Phi}}_{{\beta }_{\sf DAC}}{\rm diag}\left( {\hat{\bf h}}_k{\hat{\bf h}_k}^{\sf H} + {\bf{R}}_{\Tilde{{\bf h}}_k} \right){\bf w}_i+\frac{{\sigma}^2}{\tau P}.
\end{align}
\end{figure*}

In this subsection, we show how to optimize the power scaling parameter $\tau$,  the Lagrange multiplier $\mu$, and the antenna selection, thereby proposing the joint  precoding, antenna selection, and power control algorithm.
Since $\tau$ is a scalar variable, we use a computationally efficient gradient descent algorithm for finding the optimal $\tau$ given ${\bar{\bf w}}$ and $\mu$.
We use the same objective function $L({\bar{{\bf w}}}, \tau, \mu)$ in \eqref{lagrangian} and calculate the partial derivative with respect to $\tau$.
Let $\delta_{\sf GD}$ denote the step size of the gradient descent algorithm.
Then, the gradient descent update is given as
\begin{align}
    \label{gradient descent}
    \tau^{(t+1)} = \tau^{(t)} + \delta_{\sf GD} \frac{\partial L(\tau; {\bar{{\bf w}}}, \mu)}{\partial \tau^{(t)}}.
\end{align}
Before deriving the partial derivative of the objective function with respect to $\tau$, we define $\Xi_1$ \eqref{eq:Xi_1} and $\Xi_2$ \eqref{eq:Xi_2} shown at the bottom of the next page for the simplicity in presentation.
Using $\Xi_1$ and $\Xi_2$, we rewrite the objective function terms $L_1({\bar{{\bf w}}}, \tau)$, $L_2({\bar{{\bf w}}}, \tau)$, and $L_3({\bar{{\bf w}}}, \tau, \mu)$ in \eqref{L1}, \eqref{L2}, and \eqref{L3}, respectively, as follows:
\begin{align}
    & L_1(\tau; {\bar{{\bf w}}}) = -a \ln \left( \sum_{k = 1}^{K} \left( 
    1 + \frac{|{\bf \hat{h}}_k^{\sf H}{\bf{\Phi}}_{{\alpha }_{\sf DAC}}^{1/2} {\bf{w}}_{0}|^2} {\Xi_1}
    \right)^{-\frac {1} {a\ln2}} \right), \\
    & L_2(\tau; {\bar{{\bf w}}}) = \frac{1}{\ln{2}} \sum_{k = 1}^{K} \ln \left( 
    1 + \frac{|{\bf \hat{h}}_k^{\sf H}{\bf{\Phi}}_{{\alpha }_{\sf DAC}}^{1/2}{\bf{w}}_{k}|^2} {\Xi_2}
    \right), \\
    & L_3(\tau; {\bar{{\bf w}}}, \mu) = - \mu \left( \frac{\tau}{\varsigma}P + \Tilde{P}_{\sf cir}({\bar{\bf w}}) - P_{\sf tot} 
    \right).
\end{align}
We then derive the partial derivatives of $L_1(\tau; {\bar{{\bf w}}})$, $L_2(\tau; {\bar{{\bf w}}})$, and $L_3(\tau; {\bar{{\bf w}}}, \mu)$ with respect to $\tau$:
\begin{align}
    \label{L1'}
    &\frac{\partial L_1(\tau; {\bar{{\bf w}}})}{\partial \tau} \\
    \nonumber
    & = \frac{\sigma^2}{\tau^2 P \ln{2}} \frac{     \sum\limits_{k = 1}^{K} \left( 
    1 + \frac{|{\bf \hat{h}}_k^{\sf H}{\bf{\Phi}}_{{\alpha }_{\sf DAC}}^{1/2} {\bf{w}}_{0}|^2} {\Xi_1}
    \right)^{-\frac {1} {a\ln2}-1} \frac{|{\bf \hat{h}}_k^{\sf H}{\bf{\Phi}}_{{\alpha }_{\sf DAC}}^{1/2} {\bf{w}}_{0}|^2}{\Xi_1^2}    }{\sum\limits_{k = 1}^{K} \left( 
    1 + \frac{|{\bf \hat{h}}_k^{\sf H}{\bf{\Phi}}_{{\alpha }_{\sf DAC}}^{1/2} {\bf{w}}_{0}|^2} {\Xi_1}
    \right)^{-\frac {1} {a\ln2}}}, 
    \\
    \label{L2'}
    & \frac{\partial L_2(\tau; {\bar{{\bf w}}})}{\partial \tau} = \frac{\sigma^2}{\tau^2 P \ln{2}} \sum\limits_{k = 1}^{K} \frac{|{\bf \hat{h}}_k^{\sf H}{\bf{\Phi}}_{{\alpha }_{\sf DAC}}^{1/2} {\bf{w}}_{k}|^2}{\Xi_2 \left( \Xi_2 + |{\bf \hat{h}}_k^{\sf H}{\bf{\Phi}}_{{\alpha }_{\sf DAC}}^{1/2} {\bf{w}}_{k}|^2 \right)}, \\
    \label{L3'}
    & \frac{\partial L_3(\tau; {\bar{{\bf w}}}, \mu)}{\partial \tau} = -\frac{\mu}{\varsigma}P.
\end{align}
Finally, using \eqref{L1'}, \eqref{L2'}, and \eqref{L3'}, $\frac{\partial L(\tau; {\bar{{\bf w}}}, \mu)}{\partial \tau}$ is calculated as
\begin{align}
    \label{gradient}
    \frac{\partial L(\tau; {\bar{{\bf w}}}, \mu)}{\partial \tau} =  \frac{\partial L_1(\tau; {\bar{{\bf w}}})}{\partial \tau} \!+ \!\frac{\partial L_2(\tau; {\bar{{\bf w}}})}{\partial \tau} \!+\! \frac{\partial L_3(\tau; {\bar{{\bf w}}}, \mu)}{\partial \tau}.
\end{align}

For each gradient descent update, we use a backtracking line search method \cite{armijo1966minimization} in our simulations to decide the proper value for $\delta_{\sf GD}$. 
We also perform thresholding for each iterative update of $\tau$, given as $\tau = \max\left\{0, \min\left\{1, \tau\right\}\right\}$, to guarantee the feasibility of $\tau$, i.e.,  $\tau \in (0,1]$.

Given the power scaling parameter $\tau$ and the normalized precoder $\bar{\bf w}$, we update $\mu$ using the bisection method.
At iteration $t_{\rm \mu}$, if the power constraint is infeasible, i.e., $\frac{\tau}{\varsigma}P + \Tilde{P}_{\sf cir}({\bar{\bf w}}) > P_{\sf tot}$, $\mu^{(t_{\rm \mu})}$ is increased by step size $\delta_{\sf bm}^{( t_{\rm \mu} )}$ to put more emphasis on the power constraint term in \eqref{eq:constraint_new_new_Blk}. 
On the contrary, if the power constraint is feasible, i.e., $\frac{\tau}{\varsigma}P + \Tilde{P}_{\sf cir}({\bar{\bf w}}) \leq P_{\sf tot}$, $\mu^{(t_{\rm \mu})}$ is decreased by the same step size. 
After each iteration of this bisection algorithm, the step size $\delta_{\sf bm}^{ ( t_{\rm \mu} )}$ is reduced by half, i.e., $\delta_{\sf bm}^{( t_{\rm \mu}+1 )} = \delta_{\sf bm}^{( t_{\rm \mu} )} / 2$, thus bisecting the search interval.

Finally, we propose a quantization-aware power-constrained antenna selection (Q-PCAS) algorithm as described in Algorithm~\ref{alg:main_EEM}.
The algorithm begins by obtaining $\tau$ using the gradient descent method based on the current values of $\bar {\bf{w}}^{(t_{\rm {\bf F}})}$ and $\mu^{(t_{\mu})}$.
Next, $\bar {\bf{w}}^{(t_{\rm {\bf F}}+1)}$ is calculated using Algorithm~\ref{alg:main_QRS} with $\bar {\bf{w}}^{(t_{\rm {\bf F}})}$, $\mu^{(t_{\mu})}$, and the recently computed $\tau^{(t_{\tau})}$ as inputs.
Once $\bar{\bf{w}}$ and $\tau$ converge, the algorithm checks which antennas have noticeable effective gain.
Let $\Tilde{\bf w}_i$ denote the $i$th row of $\bf W$, and $\hat{\bf w}_i$ denote $\hat{\bf w}_i = \Tilde{\bf w}_i / \max_j \lVert \Tilde{\bf w}_j / \sqrt{\alpha_{{\sf DAC}, j}} \rVert.$
Then, only the antennas which satisfy $\lVert \hat{\bf w}_i / \sqrt{\alpha_{{\sf DAC}, i}} \rVert^2 \geq \epsilon_{\sf as}$ are selected to be activated, where $\epsilon_{\sf as} > 0$ is a sufficiently small value.
The other antennas are turned off accordingly, i.e., $\Tilde{{\bf f}}_i^{(t_{\rm {\bf F}})} = {\bf 0}_{1 \times (K+1)}$ if  $\lVert \hat{\bf w}_i / \sqrt{\alpha_{{\sf DAC}, i}} \rVert^2 < \epsilon_{\sf as}$,
where $\Tilde{{\bf f}}_i^{(t_{\rm {\bf F}})}$ denotes the $i$th row of ${\bf F}^{(t_{\rm {\bf F}})}$.
Note that this selection strategy works because the rows of $\bf W$ are jointly designed with each other by using the antenna power approximation in \eqref{eq:antenna power3}, allowing the designed $\bf W$ to incorporate relative gains across the antennas. 
Lastly, $\mu^{(t_{\mu})}$ is optimized according to the previously mentioned bisection method in the outer loop. 
After this entire procedure is repeated for $t_{\rm \mu, max}$ iterations, we recalculate the final value of $\tau$ such that the transmit power is maximized while satisfying the power constraints for the set of selected antennas.
Then, a readjustment of the normalized precoder $\bar {\bf w}$ is performed again for the recalculated $\tau$ with $\mu =0 $, using Algorithm~\ref{alg:main_QRS}.
Ideally, the solution we obtain using Algorithm 2 is a stationary point that satisfies the KKT conditions.

\begin{algorithm} [!t]
\caption{Q-PCAS} \label{alg:main_EEM} 
{\bf{initialize}}: $\bar {\bf{w}}^{(0)}, \tau^{(0)}, \text{and} \; \mu^{(0)}$\\
Set the iteration count $t_{\rm \mu} = 0$\\
\While {$t_{\rm \mu} < t_{\rm \mu, max}$}{
    Set the iteration count $t_{\rm {\bf F}} = 0$\\
    \While{$\left\| {\bf F}^{(t_{\rm {\bf F}}+1)} \! - \! {\bf F}^{(t_{\rm {\bf F}})} \right\|_F \! \mathbin{/} \! \left\| {\bf F}^{(t_{\rm {\bf F}})} \right\|_F > \epsilon_{\rm {\bf F}}$
    $\it{\&}$ $t_{\rm {\bf F}} < t_{\rm {\bf F}, max}$} {
        Set the iteration count $t_{\tau} = 0$\\
        \While {$\left| \tau^{(t_{\tau}+1)} - \tau^{(t_{\tau})} \right| \mathbin{/} \left| \tau^{(t_{\tau})} \right| > \epsilon_{\rm \tau}$ $\it{\&}$ $t_{\tau} < t_{\tau, \rm max}$} {
            Set $\frac{\partial L(\tau; {\bar{{\bf w}}}, \mu)}{\partial \tau}$ according to \eqref{gradient} \\
            Update $\tau^{(t_{\tau})}$ according to \eqref{gradient descent} \\
            $\tau^{(t_{\tau})} = \max\left\{0, \min\left\{\tau^{(t_{\tau})}, 1 \right\}\right\}$ \\
            $t_{\tau} \leftarrow t_{\tau} + 1$
        }
        $\bar {\bf{w}}^{(t_{\rm {\bf F}})} = \text{Alg.~\ref{alg:main_QRS}}\left( \bar {\bf{w}}^{(t_{\rm {\bf F}})}, \tau^{(t_{\tau})}, \mu^{(t_{\rm \mu})} \right)$ \\
        Compute ${\bf F}^{(t_{\rm {\bf F}})} = \sqrt{\tau^{(t_{\tau})}} {{\bf{\Phi}}_{{\alpha }_{\sf DAC}}^{-1/2}} \left[{\bf w}_0,{\bf w}_1, \cdots, {\bf w}_K\right]$ \\
        $t_{\rm {\bf F}} \leftarrow t_{\rm {\bf F}} + 1$
    }
    Compute $\hat{\bf W} = \Tilde{\bf W} \mathbin{/} \max_j \left\lVert \Tilde{\bf w}_j / \sqrt{\alpha_{{\sf DAC}, j}} \right\rVert$ \\
    Set $\Tilde{{\bf f}}_i^{(t_{\rm {\bf F}})} = {\bf 0}_{1 \times (K+1)}$ if $\left\lVert \hat{\bf w}_i / \sqrt{\alpha_{{\sf DAC}, i}} \right\rVert^2 < \epsilon_{\sf as}$ for $i = 1, \cdots, N$ \\
    
    Update $\mu^{(t_{\rm \mu})}$ using the bisection method \\
    $t_{\rm \mu} \leftarrow t_{\rm \mu} + 1$
}

Update $\tau = \max \left\{0,  \min \left\{ \frac{\varsigma}{P} \left( P_{\sf tot} - P_{\sf cir} \right), 1 \right\} \right\}$

Compute $\bar {\bf w} = \text{Alg.~\ref{alg:main_QRS}}\left( \bar {\bf{w}}, \tau, \mu =0\right)$ \\
Compute ${\bf F} = \sqrt{\tau} {{\bf{\Phi}}_{{\alpha }_{\sf DAC}}^{-1/2}} \left[{\bf w}_0,{\bf w}_1, \cdots, {\bf w}_K\right]$

\Return{\ }{${\bf F}$}
\end{algorithm}

\subsection{Complexity Analysis}

Let us first analyze the computational complexity of Algorithm~\ref{alg:main_QRS}.
The main computational burden for the algorithm is on the inversion of ${\bf{B}}_{\sf KKT}(\bar {\bf{w}}; \tau, \mu)$.
Since the complexity of the calculation of the inverse of a matrix of size $N \times N$ is $\mathcal{O}(N^3)$, the complexity of computing ${\bf B}_{\sf KKT}^{-1}({\bar{{\bf w}}}; \tau, \mu)$ is $\mathcal{O} (N^3K)$ because it is a block diagonal matrix that comprises of $(K+1)$ sub-matrices of size $N \times N$.
Therefore, if we denote the number of iterations in Algorithm~\ref{alg:main_QRS} as $T_{\sf GPI}$,
the complexity of the algorithm is $\mathcal{O}(T_{\sf GPI} \cdot N^3K).$

Now let us consider Algorithm~\ref{alg:main_EEM}.
For every iteration of $\tau$, the computational complexity is $\mathcal{O} ( {\mathrm {max}} ( N^2K^2,T_{\mathrm{BT}}K ) )$, 
where $T_{\mathrm{BT}}$ denotes the number of backtracking iterations.
Here, $T_{\mathrm{BT}}$ may vary for each gradient descent step.
However, for the ease of analysis, $T_{\mathrm{BT}}$ is assumed to be constant and the same for every gradient descent step.
Then, the complexity for the gradient descent algorithm is $ \mathcal{O} ( T_{\mathrm{GD}} \cdot {\mathrm {max}} ( N^2K^2,T_{\mathrm{BT}}K ) ),$ 
where $T_{\mathrm{GD}}$ denotes the number of gradient descent steps.
Finally, the overall complexity of Algorithm~\ref{alg:main_EEM} is
\begin{align}
    \mathcal{O} (
        t_{\rm \mu, max} T_{\rm {\bf F}} \cdot {\mathrm {max}} (
            T_{\sf GPI} N^3K,
            T_{\mathrm{GD}} \cdot {\mathrm {max}} ( N^2K^2,T_{\mathrm{BT}}K )
        )
    ),
\end{align}
where $T_{\rm {\bf F}}$ denotes the number of iterations of the precoder optimization.
For MIMO systems where $N\gg K$, it is reasonable to assume that $N^2K^2 > T_{\mathrm{BT}}K$ and $T_{\sf GPI} N^3K > T_{\mathrm{GD}} N^2K^2$.
As a result, the overall complexity can be simplified to $\mathcal{O} (t_{\rm \mu, max} T_{\rm {\bf F}} T_{\sf GPI} N^3K)$.

\subsection{Complexity-Reduction Approach}

\begin{figure*}[t]
    \centering
\begin{align}
    \label{eq:D_sub_diag}
    &{\bf D}_{{\sf diag}, k'}^{\sf sub} \\
    & \! \! = \! \nonumber
    \begin{cases}
        \! \sum\limits_{i = 1}^{N} \! c_{1, i} \bigl({\bf I}_{N} \! + \! \rho^{-1} \Tilde{\bf e}_i \Tilde{\bf e}_i^{\sf H}\bigr)
        \! + \! \! \sum\limits_{k = 1}^{K} \! c_{2, k}
        \Bigl(
          \sigma_{\Tilde{{\bf h}}_k}^2 \! {\bf\Phi}_{{\alpha }_{\sf DAC}} \!
          + {\bf\Phi}_{{\beta }_{\sf DAC}} \mathrm{diag}\bigl({\hat{\bf h}}_k{\hat{\bf h}}^{\sf H}_k \! + \! {\bf R}_{\Tilde{{\bf h}}_k} \! \bigr)
          \! + \! {\bf I}_{N} \! \frac{\sigma^2}{\tau P}
        \Bigr)
        \! + \! \! \sum\limits_{k = 1}^{K} \! c_{3, k}
        \Bigl(
          {\bf\Phi}_{{\beta }_{\sf DAC}} \mathrm{diag}\bigl({\hat{\bf h}}_k{\hat{\bf h}}^{\sf H}_k \! + \! {\bf R}_{\Tilde{{\bf h}}_k} \! \bigr)
          \! + \! {\bf I}_{N} \! \frac{\sigma^2}{\tau P}
        \Bigr)
        , & k' = 1, \\
        \sum\limits_{i = 1}^{N} c_{1, i}\bigl({\bf I}_{N} + \rho^{-1} \Tilde{\bf e}_i \Tilde{\bf e}_i^{\sf H}\bigr)
        + \sum\limits_{k = 1}^{K} (c_{2, k} + c_{3, k})
        \Bigl(
          \sigma_{\Tilde{{\bf h}}_k}^2 {\bf\Phi}_{{\alpha }_{\sf DAC}}
          + {\bf\Phi}_{{\beta }_{\sf DAC}} \mathrm{diag}\bigl({\hat{\bf h}}_k{\hat{\bf h}}^{\sf H}_k + {\bf R}_{\Tilde{{\bf h}}_k}\bigr)
          + {\bf I}_{N} \frac{\sigma^2}{\tau P}
        \Bigr)
        , & k' \neq 1.
    \end{cases}
\end{align}
   \vspace{-1 em}
    \\  
   \rule{\textwidth}{0.7pt}
\end{figure*}

Here, we further propose a low-complexity algorithm, denoted as Q-PCAS (Low), by using an iterative approach so that the complexity of computing ${\bf B}_{\sf KKT}^{-1}({\bar{{\bf w}}}; \tau, \mu)$ in \eqref{eq:PowerIteration} is reduced from $\mathcal{O} (N^3 K)$ to $\mathcal{O} (N^2 K)$.
For the sake of presentation, we first simplify the expression for ${\bf B}_{\sf KKT}({\bar{{\bf w}}}; \tau, \mu)$ in \eqref{eq:lem_B_kkt_Q} as
\begin{align}
    {\bf B}_{\sf KKT}({\bar{{\bf w}}}; \tau, \mu) 
    & = \sum\limits_{i = 1}^{N} c_{1, i} {\bf{E}}_i
    + \sum\limits_{k = 1}^{K} c_{2, k}  {\bf{B}}_{{\sf c},k}
    + \sum\limits_{k = 1}^{K} c_{3, k}  {\bf{B}}_k,
\end{align}
where $c_{1, i} $, $c_{2, k} $, and $c_{3, k} $ are interpreted as scalar weights for ${\bf{E}}_i$, ${\bf{B}}_{{\sf c},k}$, and ${\bf{B}}_k$, respectively.
If we assume the channels to be IID, we can approximate the covariance matrix of the channel estimation error in FDD systems \cite{wagner2012large, choi2019GPI} as ${\bf{R}}_{\Tilde{{\bf h}}_k} \approx {\bf{R}}_{\Tilde{{\bf h}}_k}^{\sf IID} = \sigma_{\Tilde{{\bf h}}_k}^2 {\bf I}$, where $\sigma_{\Tilde{{\bf h}}_k}^2 = \rho_k (2 - 2 \sqrt{1-\kappa^2})$. Here, $\rho_k$ is the large-scale fading of user $k$.
Meanwhile, recall that ${\bf B}_{\sf KKT}({\bar{{\bf w}}}; \tau, \mu)$ is a block diagonal matrix that comprises of $(K+1)$ sub-matrices of size $N \times N$.
The $(k')$th sub-matrix of ${\bf B}_{\sf KKT}({\bar{{\bf w}}}; \tau, \mu)$, denoted as ${\bf B}_{{\sf KKT}, k'}^{\sf sub} ({\bar{{\bf w}}}; \tau, \mu)$, can be expressed as
\begin{align}
    \label{eq:sub-matrix_of_Bkkt_approximated}
    &{\bf B}_{{\sf KKT}, k'}^{\sf sub} ({\bar{{\bf w}}}; \tau, \mu) \\
    &\approx 
    \nonumber
    \begin{cases} 
        {\bf D}_{{\sf diag}, k'}^{\sf sub}, & k' = 1, \\
        {\bf D}_{{\sf diag}, k'}^{\sf sub}
        + \sum\limits_{k = 1}^{K} (c_{2, k} + c_{3, k}) \left( {{\bf{\Phi}}_{{\alpha }_{\sf DAC}}^{1/2}}\right)^{\sf H} \! \hat{{\bf h}}_k \hat{{\bf h}}^{\sf H}_k {\bf{\Phi}}_{{\alpha }_{\sf DAC}}^{1/2} \\
        \quad \quad \quad - c_{3, k'-1} \left( {{\bf{\Phi}}_{{\alpha }_{\sf DAC}}^{1/2}}\right)^{\sf H} \! \hat{{\bf h}}_{k'-1} \hat{{\bf h}}^{\sf H}_{k'-1} {\bf{\Phi}}_{{\alpha }_{\sf DAC}}^{1/2}
        , & k' \neq 1,
    \end{cases}
\end{align}
where ${\bf D}_{{\sf diag}, k'}^{\sf sub}$ is a diagonal matrix which is derived in \eqref{eq:D_sub_diag} shown at the top of this page.
The complexity of inverting the first sub-matrix, ${\bf B}_{{\sf KKT}, 1}^{\sf sub} ({\bar{{\bf w}}}; \tau, \mu)$, is $\mathcal{O} (N)$ because it is approximated as ${\bf D}_{{\sf diag}, 1}^{\sf sub}$ which is a diagonal matrix of size $N \times N$.
However, for $k' \neq 1$, ${\bf B}_{{\sf KKT}, k'}^{\sf sub} ({\bar{{\bf w}}}; \tau, \mu)$ is approximated as ${\bf D}_{{\sf diag}, k'}^{\sf sub}$ with the addition of rank-one matrix terms, as can be seen in \eqref{eq:sub-matrix_of_Bkkt_approximated}.

Now, we use the Sherman-Morrison (SM) formula
$( \mathbf{A} + \mathbf{u} \mathbf{v}^{\sf H} )^{-1}
= \mathbf{A}^{-1}
- \frac{\mathbf{A}^{-1} \mathbf{u} \mathbf{v}^{\sf H} \mathbf{A}^{-1}}
       {1 + \mathbf{v}^{\sf H} \mathbf{A}^{-1} \mathbf{u}}$
\cite{sherman1950adjustment} 
to recursively compute $( {\bf B}_{{\sf KKT}, k'}^{\sf sub} ({\bar{{\bf w}}}; \tau, \mu) )^{-1}$ for $k' = 2, \cdots, K+1$.
For $k' \neq 1$ and $l = 1, \cdots, K$, we define an auxiliary matrix ${\bf Z}_{\sf SM}^{(l)}$:
\vspace{-0.5em}
\begin{align}
    {\bf Z}_{\sf SM}^{(l)} = {\bf D}_{{\sf diag}, k'}^{\sf sub} + \sum_{k = 1}^{l} (c_{2, k} + c_{3, k}) \mathbf{u}_k \mathbf{u}_k^{\sf H},
\end{align}
where $\mathbf{u}_k = ( {{\bf{\Phi}}_{{\alpha }_{\sf DAC}}^{1/2}} )^{\sf H}  \hat{{\bf h}}_k$.
We note that ${\bf Z}_{\sf SM}^{(l)}$ does not depend on $k'$.
Using the SM formula, $( {\bf Z}_{\sf SM}^{(K)} )^{-1}$ can be recursively computed as
\begin{align}
    &\left( {\bf Z}_{\sf SM}^{(k)} \right)^{-1}
    = \left( {\bf Z}_{\sf SM}^{(k-1)} + (c_{2, k} + c_{3, k}) \mathbf{u}_k \mathbf{u}_k^{\sf H} \right)^{-1} \\
    \nonumber
    &= \left( {\bf Z}_{\sf SM}^{(k-1)} \right)^{-1} 
        - \frac{ (c_{2, k} + c_{3, k}) \left( {\bf Z}_{\sf SM}^{(k-1)} \right)^{-1} \mathbf{u}_k \mathbf{u}_k^{\sf H}  \left( {\bf Z}_{\sf SM}^{(k-1)} \right)^{-1} }{ 1 + (c_{2, k} + c_{3, k}) \mathbf{u}_k^{\sf H}  \left( {\bf Z}_{\sf SM}^{(k-1)} \right)^{-1}  \mathbf{u}_k  },
\end{align}
where ${\bf Z}_{\sf SM}^{(0)} = {\bf D}_{{\sf diag}, k'}^{\sf sub}$.
We remark that the complexity of recursively computing $( {\bf Z}_{\sf SM}^{(K)} )^{-1}$ is $\mathcal{O} (N^2 K)$.
Finally, $( {\bf B}_{{\sf KKT}, k'}^{\sf sub} ({\bar{{\bf w}}}; \tau, \mu) )^{-1}$ for $k' = 2, \cdots \!, K+1$ can be computed as
\begin{align}
    \nonumber
    & \left( {\bf B}_{{\sf KKT}, k'}^{\sf sub} ({\bar{{\bf w}}}; \tau, \mu) \right)^{-1}
    = \left(  {\bf Z}_{\sf SM}^{(K)}  
                - c_{3, k'-1} \mathbf{u}_{k'-1} \mathbf{u}_{k'-1}^{\sf H}
        \right)^{-1} \\
    & =  \left( {\bf Z}_{\sf SM}^{(K)} \right)^{-1}
        + \frac{ c_{3, k'-1} \left( {\bf Z}_{\sf SM}^{(K)} \right)^{-1} \mathbf{u}_{k'-1} \mathbf{u}_{k'-1}^{\sf H} \left( {\bf Z}_{\sf SM}^{(K)} \right)^{-1} }{ 1 - c_{3, k'-1} \mathbf{u}_{k'-1}^{\sf H} \left( {\bf Z}_{\sf SM}^{(K)} \right)^{-1} \mathbf{u}_{k'-1} }.
\end{align}
Since $( {\bf Z}_{\sf SM}^{(K)} )^{-1}$ is already computed, the additional complexity of computing $( {\bf B}_{{\sf KKT}, k'}^{\sf sub} ({\bar{{\bf w}}}; \tau, \mu) )^{-1}$ is $\mathcal{O} (N^2)$.
Hence, by using the SM formula, we can drastically reduce the complexity of computing $\left( {\bf B}_{\sf KKT}({\bar{{\bf w}}}; \tau, \mu) \right)^{-1}$ from $\mathcal{O} (N^3 K)$ to $\mathcal{O} (N^2 K)$. 
Accordingly, the complexity of our low-complexity algorithm is reduced from $\mathcal{O} (t_{\rm \mu, max} T_{\rm {\bf F}} T_{\sf GPI} N^3K)$ to $\mathcal{O} (t_{\rm \mu, max} T_{\rm {\bf F}} T_{\sf GPI} N^2K)$.

\section{Numerical Results}

In our simulations, we utilize the one-ring model \cite{adhi:tit:13} to generate the channel vector of user $k$: $ {\bf h}_k = \sqrt{\rho_k}{\bf g}_k.$
Here, $\rho_k$ and ${\bf g}_k$ denote the large and small-scale fading of user $k$, respectively.
The pathloss $\rho_k$ is generated using the log-distance pathloss model described in \cite{erceg1999empirically}.
We also set $\kappa = 0.4$ unless specified otherwise.
This is so that the estimation error power is about $16\%$ of the channel power.
The rest of the system parameters are set as follows: a cell radius of 1 km, a minimum distance of 100 m between the BS and users, a pathloss exponent of 4, a carrier frequency of 2.4 GHz with a bandwidth of 150 MHz (passband), a lognormal shadowing variance of 8.7 dB, and a noise figure of 5 dB.
The angle-of-departure (AoD) of user $k$, denoted as $\theta_k$, is randomly distributed such that the maximum AoD difference of users is $0.1$ radians.
Meanwhile, we set
$P_{\sf LO} = 22.5$ mW,
$P_{\sf LP} = 14$ mW,
$P_{\sf M} = 0.3$ mW,
$P_{\sf H} = 3$ mW,
$\varsigma = 0.27$,
and the sampling rate as $f_s = 150$ MHz.
Unless specified otherwise, we use $b_{{\sf DAC}, i} \in \{4, 8, 12, 16\}$ quantization bits, each of which is equally assigned to $N/4$ antennas.
Lastly, the parameters used in Algorithm \ref{alg:main_QRS} and \ref{alg:main_EEM} are set as follows: $\epsilon = 0.01$, $t_{\rm max} = 20$, $\rho = 10^{-12}$, $t_{\rm \mu, max} = 20$, $\epsilon_{\rm {\bf F}} = 0.01$, $t_{\rm {\bf F}, max} = 20$, $\epsilon_{\rm \tau} = 0.001$, $t_{\rm \tau, max} = 20$, and $\epsilon_{\sf as} = 0.1$, unless mentioned otherwise.
In addition, $\tau$ and $\mu$ are initialized as $\tau = 1$ and $\mu = 0$, and the precoder is initialized using the Q-RZF algorithm introduced below.
The initial step sizes for the backtracking line search algorithm and the bisection algorithm are set as $\delta_{\sf GD} = 1$ and $\delta_{\sf bm} = 1$, respectively, unless mentioned otherwise.
We remark that the converged value of $\mu$ has always been between $0$ and $\delta_{\sf bm}$ in our simulations.
As for $a$ in \eqref{eq:approxcom}, its values were chosen large enough for Algorithm~\ref{alg:main_QRS} to converge. 

\begin{figure}[!t]\centering
    \begin{subfigure}
        [Sum spectral efficiency]{\resizebox{0.84\columnwidth}{!}{\includegraphics{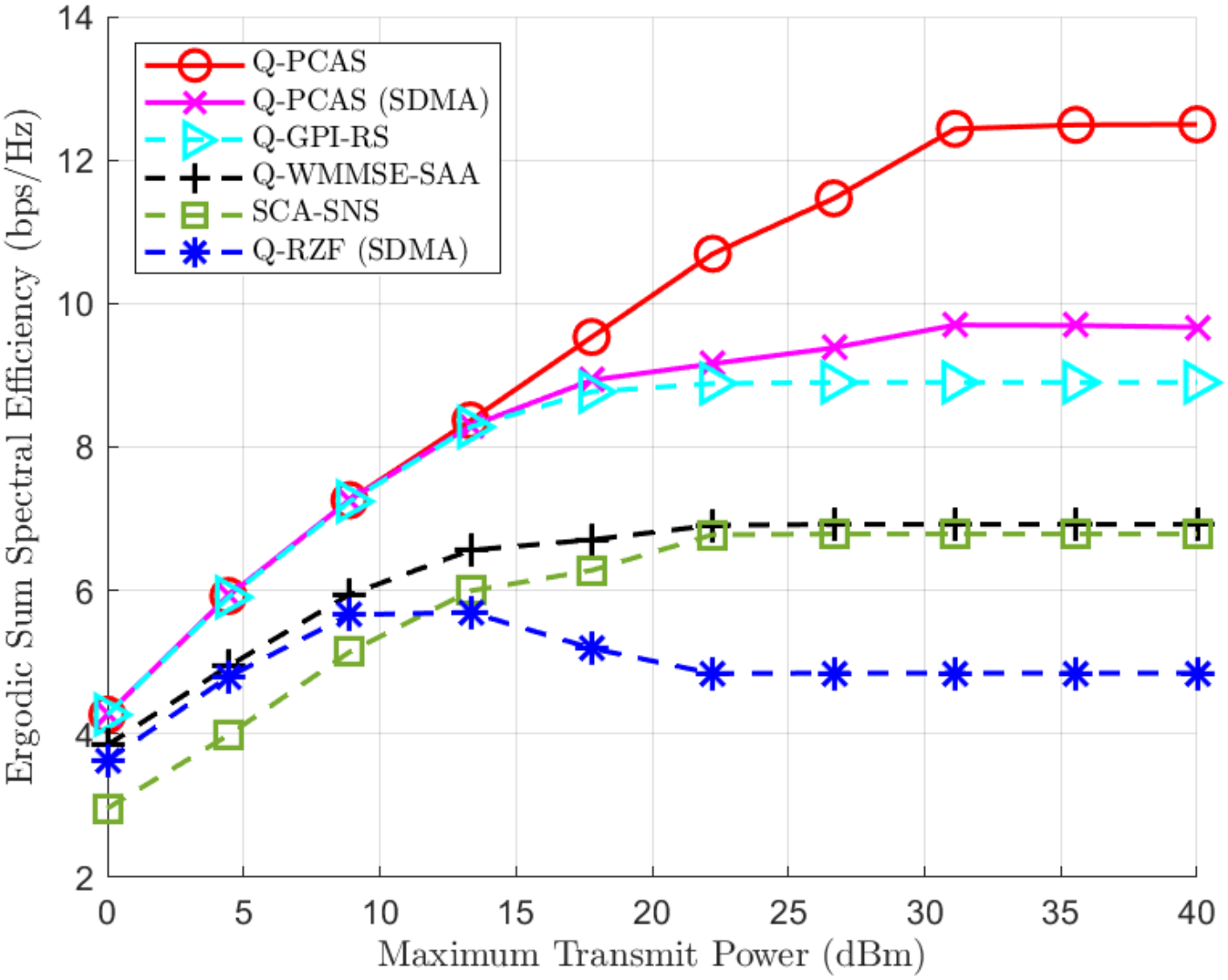}}}
    \end{subfigure}
    \begin{subfigure}
        [Selection ratio]{\resizebox{0.84\columnwidth}{!}{\includegraphics{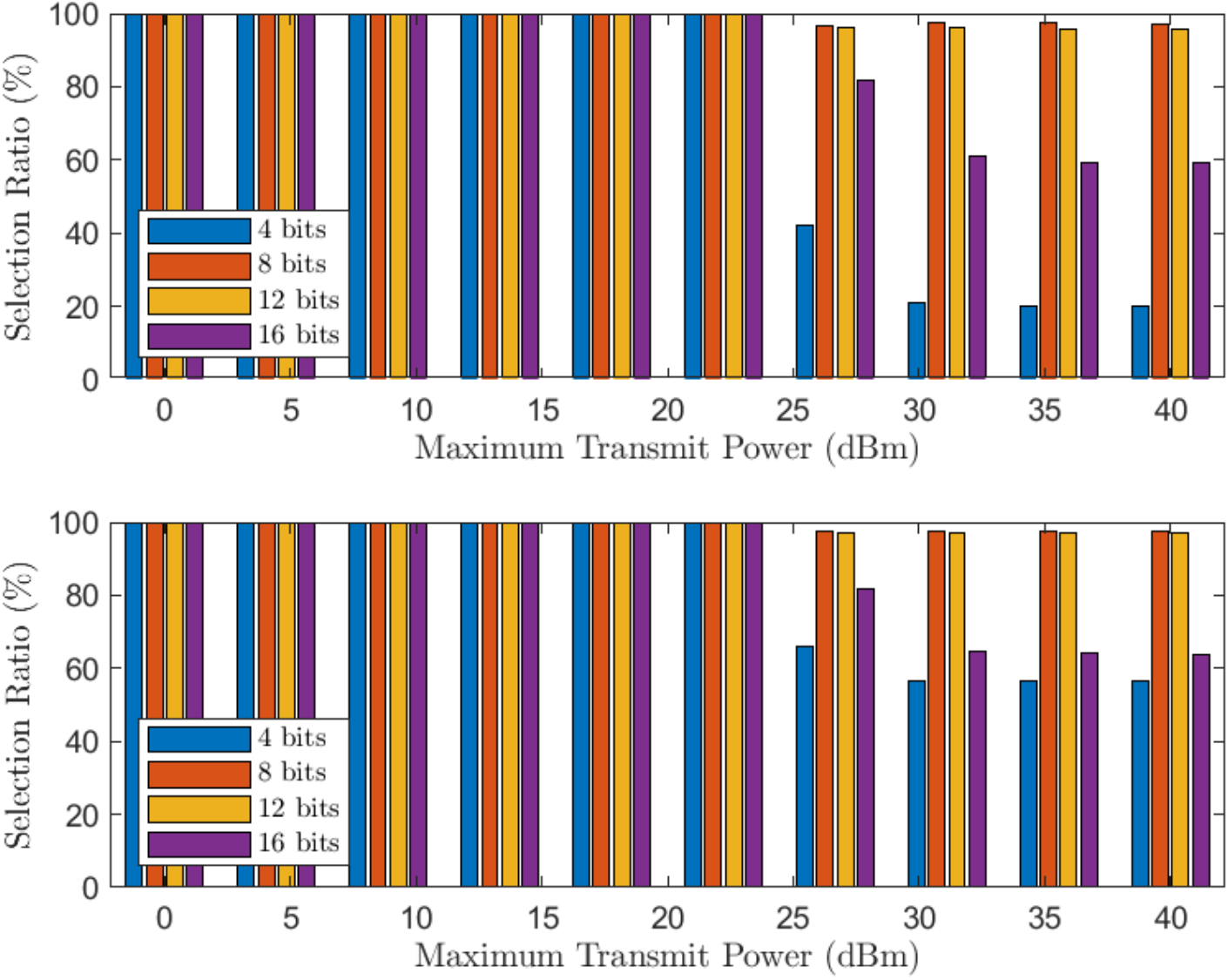}}}
    \end{subfigure}
    \begin{subfigure}
        [SE of the common and private streams]{\resizebox{0.84\columnwidth}{!}{\includegraphics{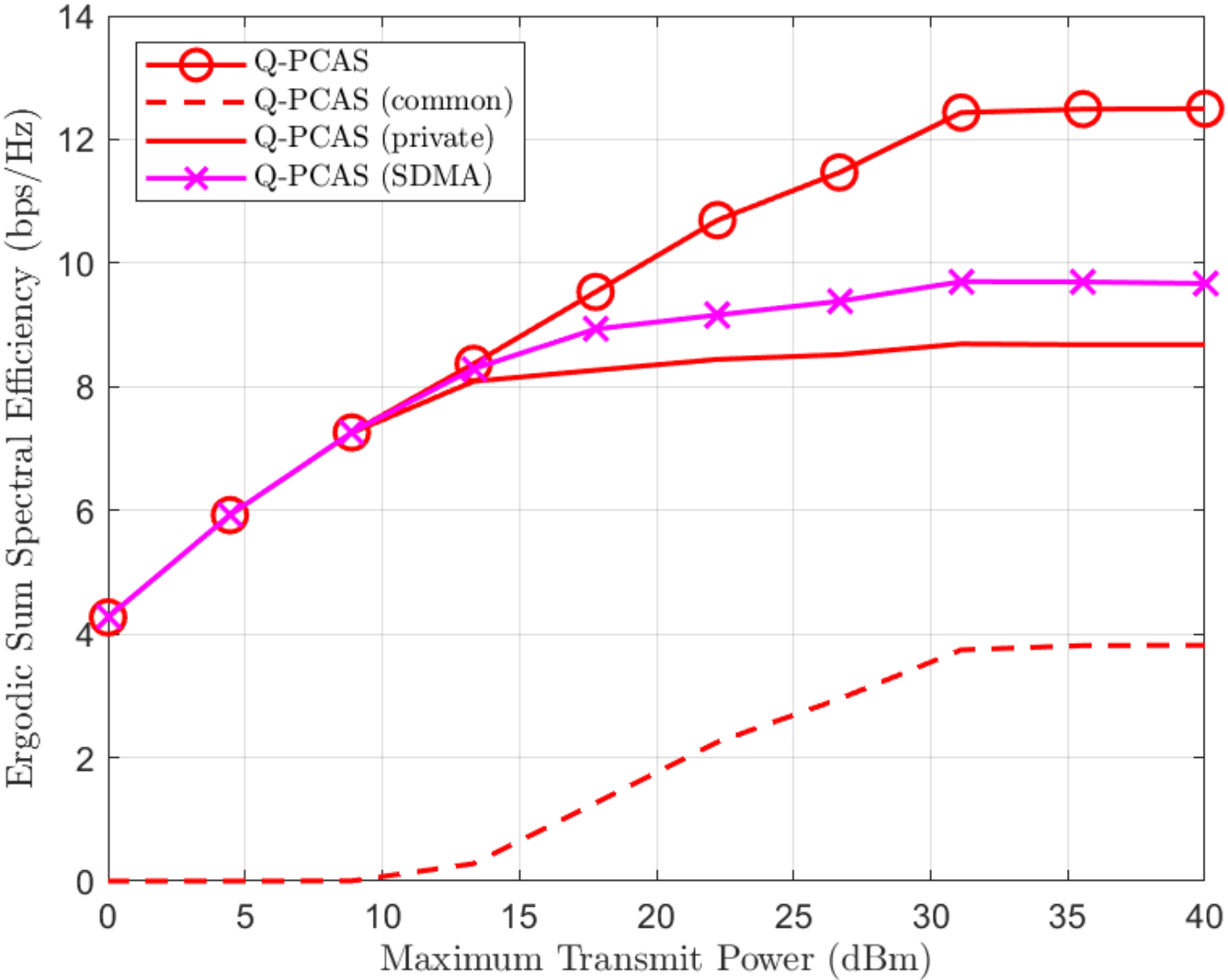}}}
    \end{subfigure}
\caption{
    (a) The sum SE with respect to the maximum transmit power $P$,
    (b) the selection ratio of each DAC resolution $b_{{\sf DAC}, i} \in \{4, 8, 12, 16\}$, 
    and (c) the SE of the common and private streams, all with respect to the maximum transmit power $P$ for $N = 16$ BS antennas, $K = 4$ users, and $P_{\sf tot} = 40$ dBm. The top and bottom bar graphs in (b) correspond to our proposed RSMA and SDMA algorithms, respectively.}
\label{fig:R_wrt_P}
\vspace{-1em}
\end{figure}

The baseline algorithms used for comparison are:
\begin{itemize}
    \item {\bf Q-GPI-RS} (RSMA): The initial work on quantized RSMA precoding under perfect CSIT using the GPI method \cite{park2022rate}.
    This algorithm only considers the maximum transmit power and does not consider the circuit power.
    The complexity order is $\mathcal{O} (T_{\sf GPI} N^3K)$.
    
    \item {\bf Q-WMMSE-SAA} (RSMA): The quantization-aware version of the weighted minimum mean square error (WMMSE) algorithm \cite{joudeh2016sum}.
    To handle imperfect CSIT, this algorithm employs the sample average approximation (SAA) approach, which obtains the empirical average over $M=1000$ channel realizations.
    The complexity order is $\mathcal{O} (T_{\sf WMMSE}N^{3.5}K^{3.5})$, where $T_{\sf WMMSE}$ is the number of iterations of the WMMSE algorithm.
    
    \item {\bf SCA-SNS} (RSMA): The successive null-space precoding (SNS) algorithm \cite{krishnamoorthy2022downlink}, which has been shown to be robust to imperfect CSIT.
    This algorithm uses linear combinations of null-space basis vectors of successively augmented channel matrices to create precoding vectors.
    Successive convex approximation (SCA) is used to tackle the non-convexity of the problem.
    The complexity order is $\mathcal{O} (T_{\sf SCA}(N^2 + K^3 + NK(N-K) )^{1.5}\log(1/\epsilon_{\sf SCA}))$, where $T_{\sf SCA}$ and $\epsilon_{\sf SCA}$ are the number of iterations and the numerical tolerance of the SCA algorithm, respectively.
    
    \item {\bf{Q-RZF}} (SDMA): The conventional quantization-aware linear regularized zero-forcing (RZF) precoder, which is constructed using the effective channel ${\hat{\bf h}}^{\rm eff}_{k} = {\bf{\Phi}}_{{\alpha }_{\sf DAC}}^{\sf H}{\hat{\bf h}_{k}}$.
    This baseline uses the space-division multiple access (SDMA) framework, a special case of RSMA when there is no message splitting.
\end{itemize}
For our baselines, we incorporate a plain power scaling approach which simply reduces the transmit power within the power budget in power-constrained regimes.

\subsection{Heterogeneous-Resolution DACs}

We first consider the case where the BS uses heterogeneous resolution DACs.
The proposed RSMA algorithm, Q-PCAS, is compared with the baselines including the SDMA equivalent of our algorithm, Q-PCAS (SDMA), in terms of the SE, where $N = 16$, $K = 4$, and $P_{\sf tot} = 40$ dBm.
We evaluate the SE performance with respect to the maximum transmit power $P$ in Fig.~\ref{fig:R_wrt_P}(a). 
As $P$ increases for fixed $P_{\sf tot} = 40$ dBm, the BS can have more flexibility in using higher transmit power with less antennas or more antennas with lower transmit power.
Accordingly, as shown in Fig.~\ref{fig:R_wrt_P}(a), the proposed Q-PCAS algorithm outperforms the baselines, capturing the best performance trade-off through joint optimization.
In addition, when comparing Q-PCAS with Q-PCAS (SDMA), the RSMA gain is noticeable thanks to its IUI mitigation ability and flexibility in its spatial degree-of-freedom.
Meanwhile, the other baselines show no SE gains after a certain value of $P$.
This is because in power-constrained regimes, the transmit power of the baselines, $P_{\sf tx} = \varsigma \left( P_{\sf tot} - P_{\sf cir} \right)$, is constant.
In contrast, Q-PCAS appropriately selects the best antennas and the transmit power for any given total power budget at the BS.

In Fig.~\ref{fig:R_wrt_P}(b), the antenna selection ratio is plotted against the maximum transmit power $P$ for the proposed RSMA (top) and SDMA (bottom) algorithms under the same system settings as Fig.~\ref{fig:R_wrt_P}.
In the power-constrained regime ($P>25$ dBm), the antennas with the highest and lowest resolution DACs are turned off first for both the RSMA and SDMA algorithms: the $16$-bit DACs are turned off with high priority because they are power inefficient,
and the $4$-bit DACs are also turned off with high priority because they contribute marginally to the SE due to the quantization error, which increases with the transmit power $P_{\sf tx}$.
In addition, when comparing the proposed SDMA algorithm to the RSMA one, less antennas are turned off for any given $P$ in the power-constrained regime.
We conjecture that this is because for SDMA, there is no common stream as in the case of RSMA, and is thus more reluctant to turn off antennas and sacrifice degree of freedom.
It accordingly chooses to have a higher channel rank by using more transmit antennas and allocate less power to the transmit power $P_{\sf tx}$.
We remark in this simulation, the transmit power of our proposed RSMA algorithm is $16.4\%$ higher than that of our SDMA algorithm in the power-constrained regime.
Meanwhile, a large portion of the SE gain achieved by the proposed RSMA algorithm is attributed to the use of the common stream as shown in Fig.~\ref{fig:R_wrt_P}(c).
This allows more room for flexibility in antenna selection when power is limited.

\subsection{Homogeneous-Resolution DACs}

\begin{figure}[!t]\centering
	\subfigure{\resizebox{0.84\columnwidth}{!}{\includegraphics{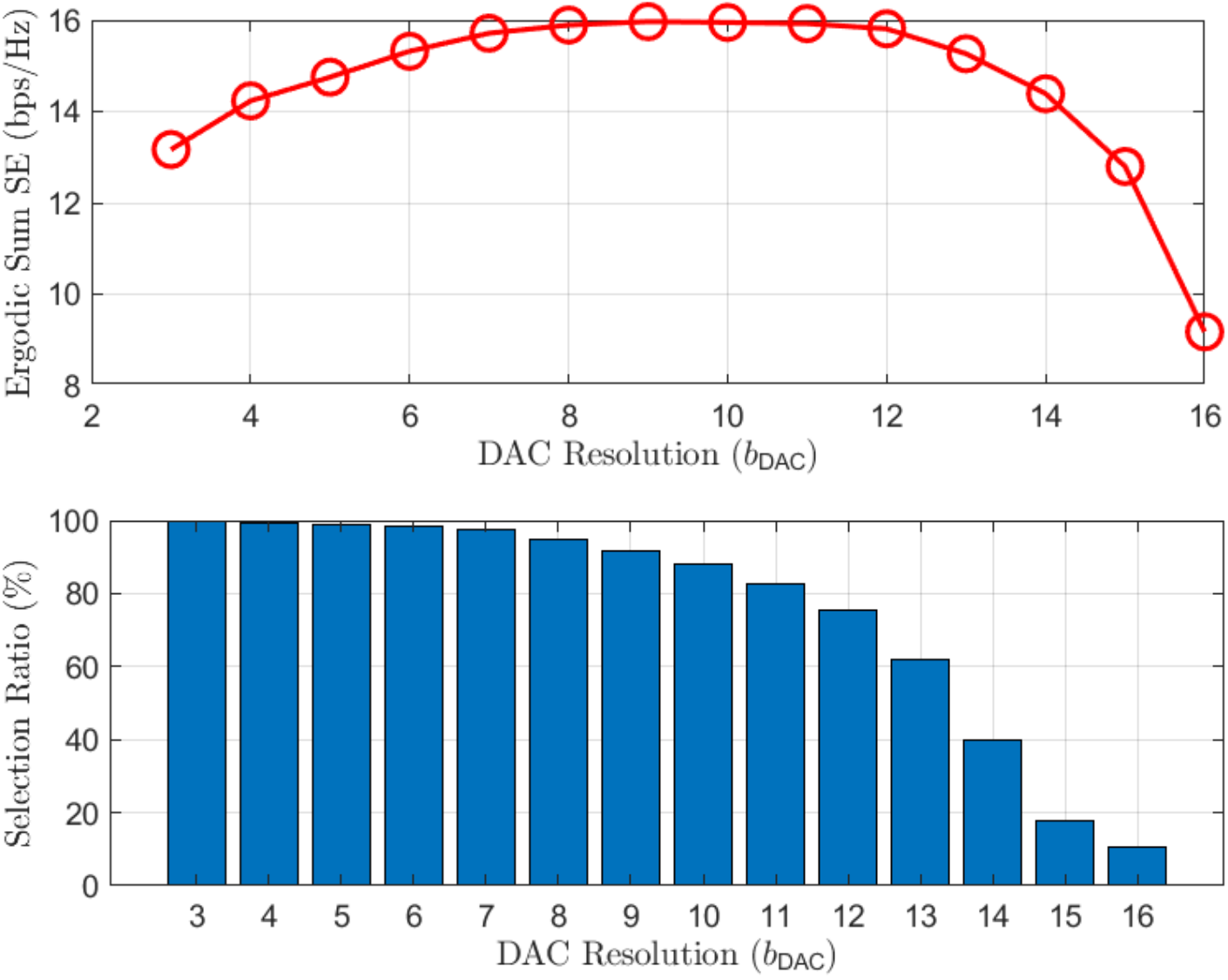}}}
\caption{
    SE with respect to the quantization resolution $b_{\sf DAC}$ (top) and the selection ratio of each resolution (bottom) for $N = 32$ BS antennas, $K = 4$ users, and $P = P_{\sf tot} = 40$ dBm.
    } 
 \label{fig:R_wrt_bDAC}
 \vspace{-1em}
\end{figure}

We now evaluate the proposed algorithm for homogeneous-resolution DAC cases.
To have a better understanding on which resolution DACs are more power-efficient than others, we plot the sum SE against the DAC resolution $b_{{\sf DAC}}$ and the corresponding antenna selection ratios in Fig.~\ref{fig:R_wrt_bDAC} for $N = 32$, $K = 4$, $P = P_{\sf tot} = 40$ dBm, and $\delta_{\sf bm} = 8$.
As shown in Fig.~\ref{fig:R_wrt_bDAC}, when both antenna selection and power scaling are taken into consideration, the most power-efficient DACs are medium-resolution DACs that use $8\sim11$ bits.
This result demonstrates that under the joint optimization of antenna selection and power control for a given BS power budget $P_{\rm tot}$, the conventional belief that low-resolution quantization bits such as $3\sim5$ bits are the most power-efficient \cite{ding2019spectral, atzeni2022lowres} is no longer valid.
It is also observed that more antennas are selected to be turned off as resolution improves because the power consumption burden on the DACs increases exponentially, which aligns with intuition.
In addition, Fig.~\ref{fig:homo_R_wrt_P} depicts the sum SE using $8$-bit resolution DACs with respect to the maximum transmit power for
$N = 16$, $K = 4$, and $P_{\sf tot} = 32$ dBm.
Similar to Fig.~\ref{fig:R_wrt_P}, the proposed algorithm achieves the highest SE.

\begin{figure}[!t]\centering
	\subfigure{\resizebox{0.84\columnwidth}{!}{\includegraphics{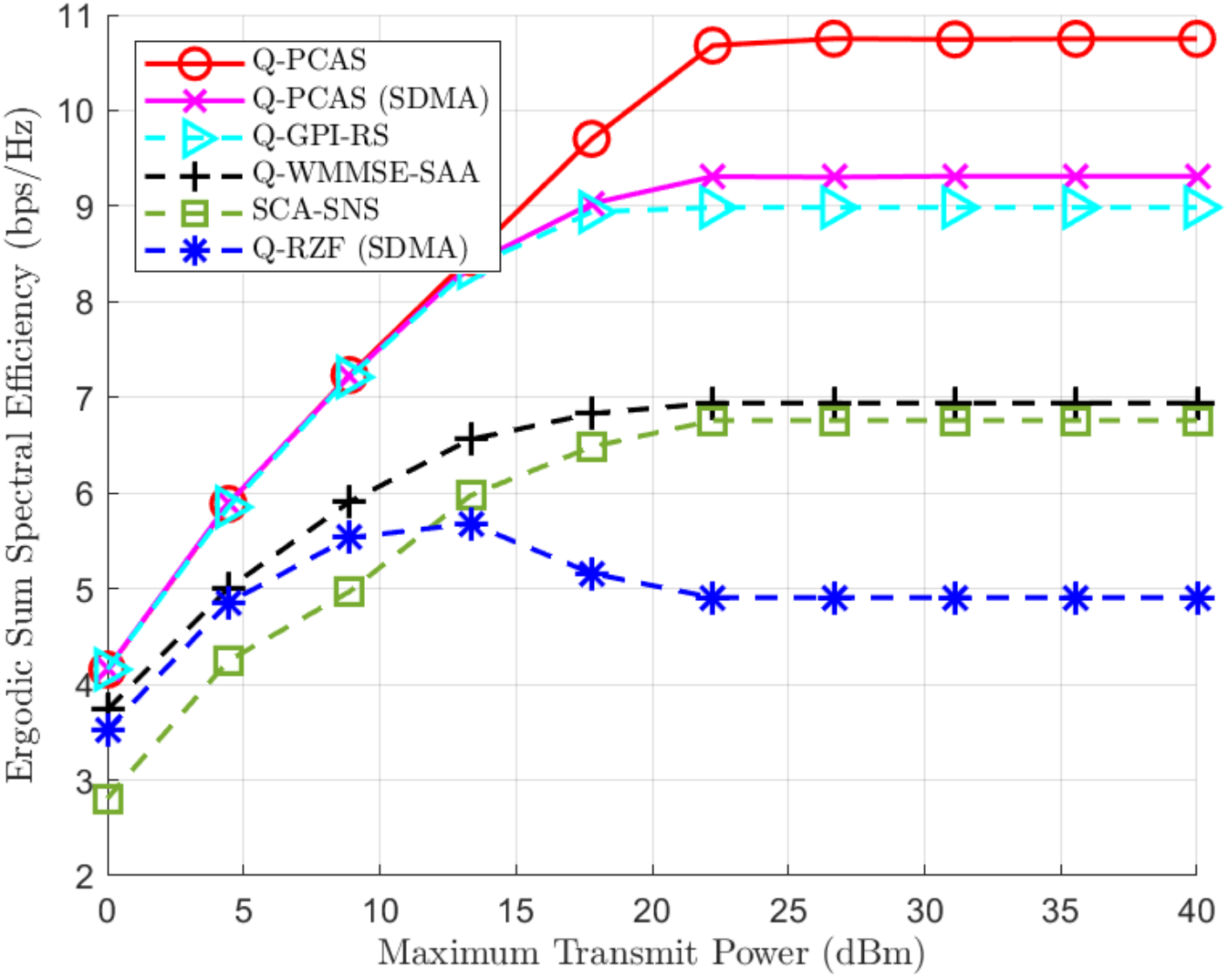}}}
\caption{
    The sum SE using homogeneous resolution DACs $(b_{\sf DAC} = 8)$ with respect to the maximum transmit power $P$ for $N = 16$ BS antennas, $K = 4$ users, and $P_{\sf tot} = 32$ dBm.} 
 \label{fig:homo_R_wrt_P}
 \vspace{-1em}
\end{figure}


\subsection{Convergence}

In Fig.~\ref{fig:convergence}, the convergence of the proposed algorithm is simulated for $N = 16$, $K = 4$, and $P = P_{\sf tot} = 40$ dBm.
The top graph in Fig.~\ref{fig:convergence} shows the normalized error $\left\| {\bf F}^{(t_{\rm {\bf F}}+1)} \! - \! {\bf F}^{(t_{\rm {\bf F}})} \right\|_F \! \mathbin{/} \! \left\| {\bf F}^{(t_{\rm {\bf F}})} \right\|_F$ with respect to each iteration $t_{\rm {\bf F}}$ of the precoder ${\bf F}$ in Algorithm~\ref{alg:main_EEM}.
It is shown that ${\bf F}$ converges within two iterations in most cases.
The middle and bottom graphs in Fig.~\ref{fig:convergence} show the convergence results in terms of the sum SE and its corresponding value of $\mu$, respectively, with respect to each iteration $t_{\rm \mu}$ of $\mu$ in Algorithm~\ref{alg:main_EEM}.
As illustrated,
the sum SE performance and its corresponding value of $\mu$ converges within a few iterations.

\begin{figure}[!t]\centering
	\subfigure{\resizebox{0.84\columnwidth}{!}{\includegraphics{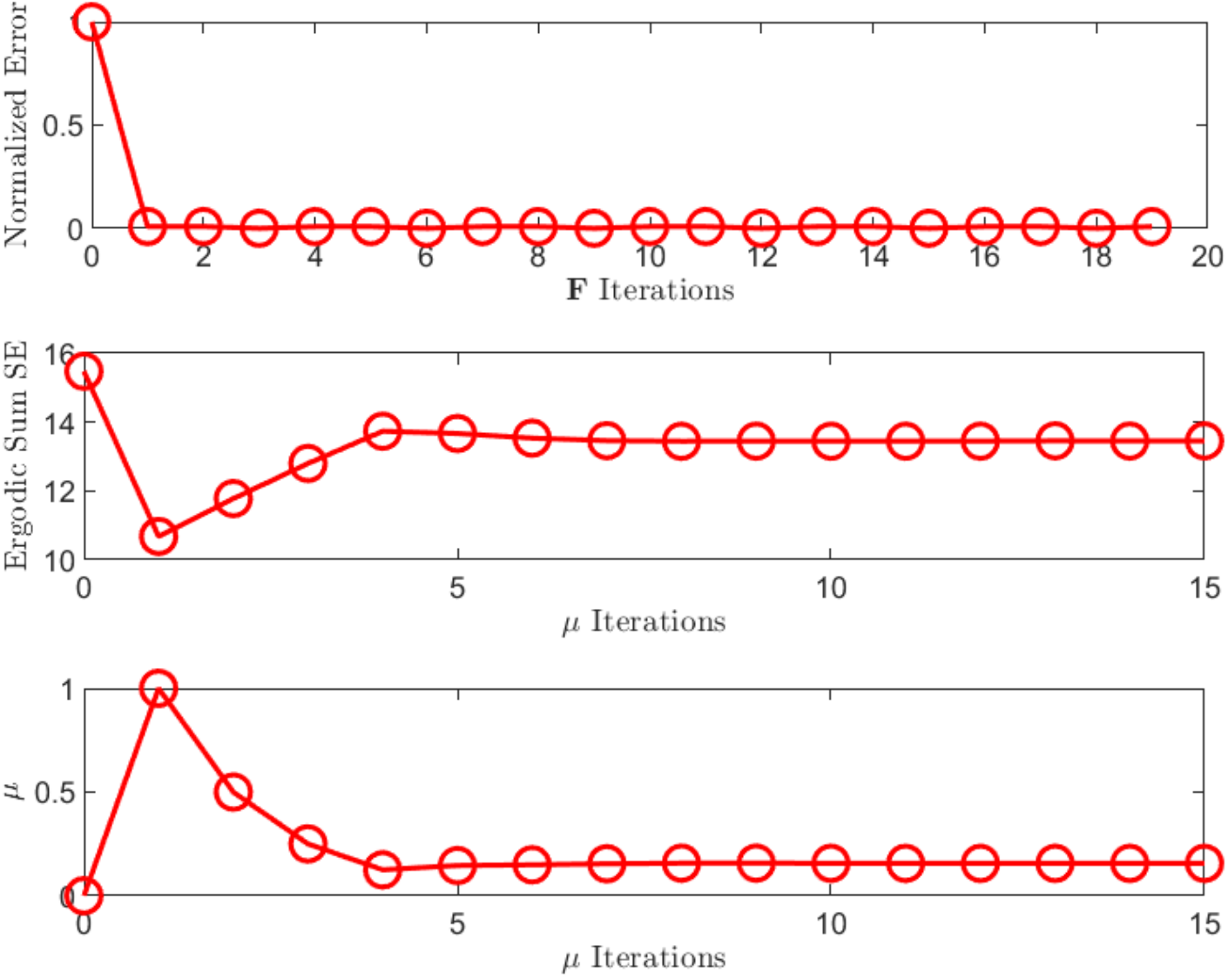}}}
\caption{
    The normalized error $\| {\bf F}^{(t_{\rm {\bf F}}+1)} \! - \! {\bf F}^{(t_{\rm {\bf F}})} \|_F \! \mathbin{/} \! \| {\bf F}^{(t_{\rm {\bf F}})} \|_F$ for each iteration $t_{\rm {\bf F}}$ of ${\bf F}$ (top), and the sum SE (middle) and its corresponding $\mu$ value (bottom) for each iteration $t_{\rm \mu}$ of $\mu$ in Algorithm~\ref{alg:main_EEM}.
    The system parameters are $N = 16$ BS antennas, $K = 4$ users, and $P = P_{\sf tot} = 40$ dBm.} 
 \label{fig:convergence}
 \vspace{-1em}
\end{figure}

\subsection{Channel Estimation Error}

\begin{figure}[!t]\centering
	\subfigure{\resizebox{0.84\columnwidth}{!}{\includegraphics{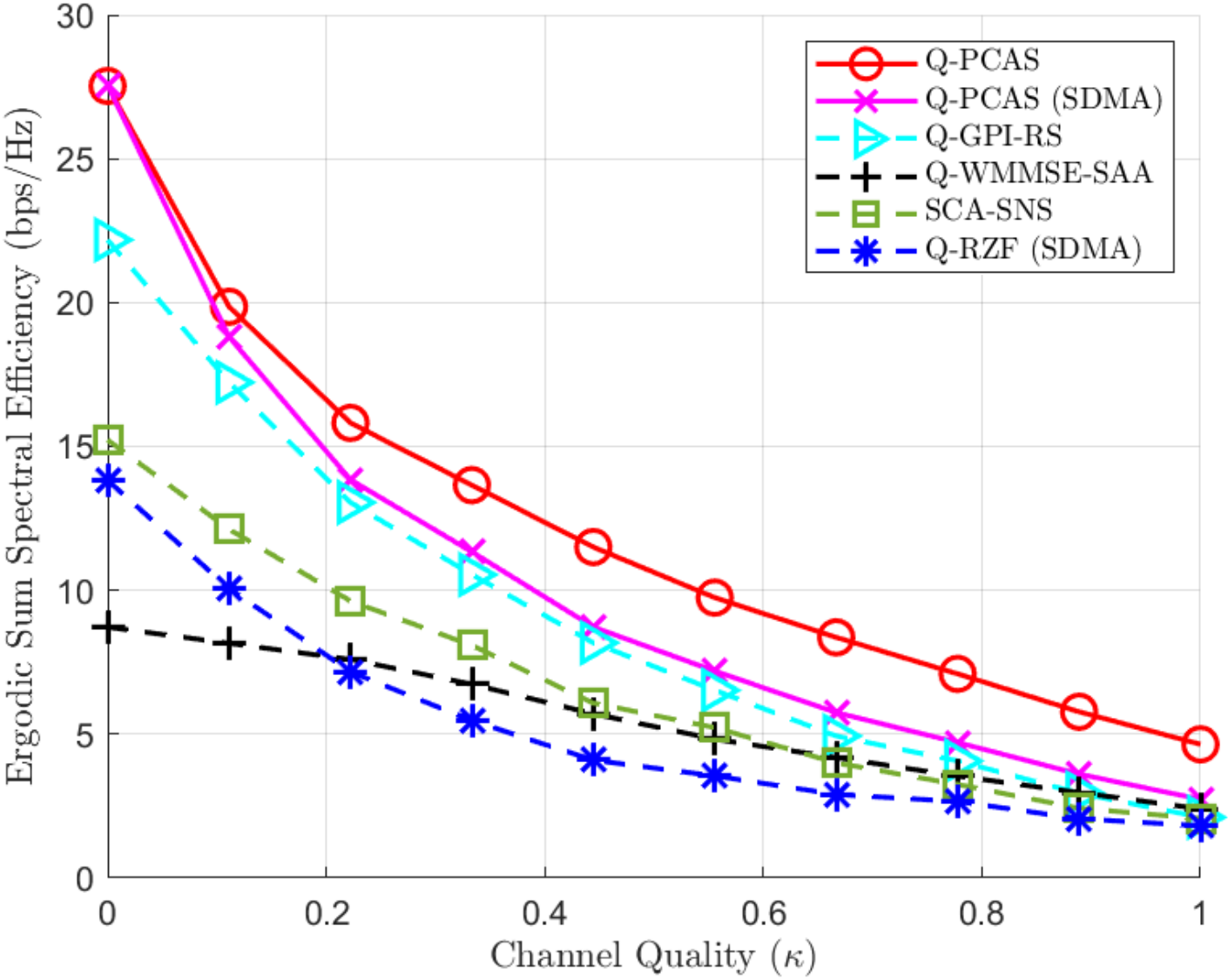}}}
\caption{
    The sum SE with respect to $\kappa \in [0, 1]$, a parameter that indicates the channel quality, for $N = 16$ BS antennas, $K = 4$ users, $P = 30$ dBm, and $P_{\sf tot} = 40$ dBm.
    } 
 \label{fig:R_wrt_kappa}
 \vspace{-1em}
\end{figure}

To evaluate the proposed algorithm with respect to the channel estimation accuracy, the sum SE is plotted against the channel quality parameter $\kappa$ for $N = 16$, $K = 4$, $P = 30$ dBm, and $P_{\sf tot} = 40$ dBm in Fig.~\ref{fig:R_wrt_kappa}.
Here, $\kappa = 0$ is when there is no channel estimation error, whereas $\kappa = 1$ is when the channel estimation error is the worst.
The proposed algorithm has the best performance out of all the baselines for imperfect CSIT, i.e., when $\kappa \neq 0$.
The SE gain for $\kappa = 1$ is attributed to ${\bf q}_k$, the CSIT error of user $k$,  having some residue information of the channel covariance matrix in the channel error covariance matrix \cite{choi2019GPI} as Q-PCAS exploits the  error covariance matrix in \eqref{eq:channelcov_R1}.

\subsection{Complexity-Reduction Approach}

\begin{figure}[!t]\centering
	\subfigure{\resizebox{0.84\columnwidth}{!}{\includegraphics{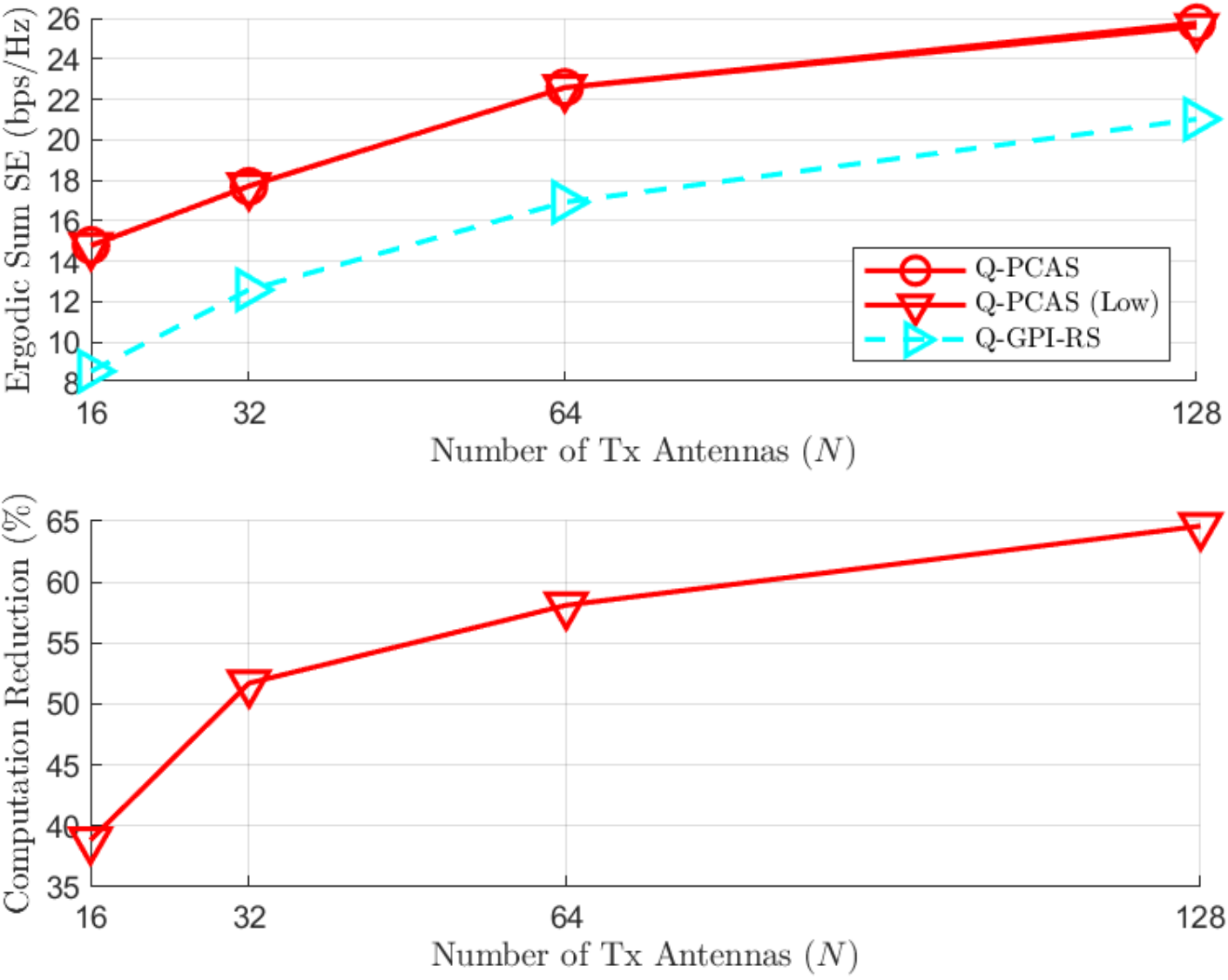}}}
\caption{
    The sum spectral efficiency (top) and the computation time reduction (bottom) of our low-complexity algorithm (in $\%$) with respect to the number of BS antennas $N$ for $K=4$ users and $P = P_{\sf tot} = 50$ dBm.
    } 
\label{fig:fast}
 \vspace{-1em}
\end{figure}

In the top subplot of Fig.~\ref{fig:fast}, we plot the sum spectral efficiency of our low-complexity algorithm with respect to $N$ for $K=4$ users, $P = P_{\sf tot} = 50$ dBm, $t_{\mu,\rm max} = 8$, and $T_{\bf F} = 10$.
Our low-complexity algorithm, denoted as Q-PCAS (Low), has negligible performance loss compared to our original proposed algorithm.
The computation time, however, is drastically reduced.
The bottom subplot of Fig.~\ref{fig:fast} plots the average computation time reduction $(1 - T_{\sf Q-PCAS (Low)}/T_{\sf Q-PCAS}) \times 100 \%$ with respect to $N$, where $T_{\sf Q-PCAS (Low)}$ and $T_{\sf Q-PCAS}$ are the average computation times of Q-PCAS (Low) and Q-PCAS, respectively.
From Fig.~\ref{fig:fast}, we can conclude that our low-complexity algorithm successfully reduces the complexity of our proposed algorithm while still achieving superior SE performance. 
Since $t_{\mu,\rm max}T_{\bf F} <N$ is typically true in massive MIMO systems, Q-PCAS (Low) has the lowest complexity order among the considered state-of-the art algorithms.

\section{Conclusion}

In this paper, we proposed an RSMA algorithm that jointly optimizes the precoder, antenna selection, and the transmit power such that sum SE performance is maximized for a total power budget at a BS.
With a series of problem reformulations, we incorporated the partial channel knowledge in the problem, expressed the problem as a GPI-friendly form, and used approximation techniques for tractability.
We also proposed a complexity-reduction approach that is more suitable for massive MIMO systems.
Numerical results show that the proposed algorithm outperforms all the baselines and that contrary to the conventional notion that low-resolution DACs with $3\sim 5$ quantization bits are the most power-efficient, medium-resolution DACs with $8\sim 11$ bits may actually be more power-efficient when utilizing the full potential of the available power at the BS.

\bibliographystyle{IEEEtran}
\bibliography{draft_final}

\begin{thebibliography}{10}
\providecommand{\url}[1]{#1}
\csname url@samestyle\endcsname
\providecommand{\newblock}{\relax}
\providecommand{\bibinfo}[2]{#2}
\providecommand{\BIBentrySTDinterwordspacing}{\spaceskip=0pt\relax}
\providecommand{\BIBentryALTinterwordstretchfactor}{4}
\providecommand{\BIBentryALTinterwordspacing}{\spaceskip=\fontdimen2\font plus
\BIBentryALTinterwordstretchfactor\fontdimen3\font minus \fontdimen4\font\relax}
\providecommand{\BIBforeignlanguage}[2]{{%
\expandafter\ifx\csname l@#1\endcsname\relax
\typeout{** WARNING: IEEEtran.bst: No hyphenation pattern has been}%
\typeout{** loaded for the language `#1'. Using the pattern for}%
\typeout{** the default language instead.}%
\else
\language=\csname l@#1\endcsname
\fi
#2}}
\providecommand{\BIBdecl}{\relax}
\BIBdecl

\bibitem{sung2025joint_TOAPPEAR}
J.~Sung, S.~Park, and J.~Choi, ``{Joint Optimization for Power-Constrained MIMO Systems: Is Low-Resolution DAC Still Optimal?}'' in \emph{2025 IEEE 101st Veh. Technol. Conf. (to appear)}, 2025.

\bibitem{mimo:twc:10:marzetta}
T.~L. {Marzetta}, ``Noncooperative cellular wireless with unlimited numbers of base station antennas,'' \emph{IEEE Trans. Wireless Commun.}, vol.~9, no.~11, pp. 3590--3600, Nov. 2010.

\bibitem{sanayei:commmag:04}
S.~{Sanayei} and A.~{Nosratinia}, ``Antenna selection in {MIMO} systems,'' \emph{IEEE Commun. Mag.}, vol.~42, no.~10, pp. 68--73, 2004.

\bibitem{tervo2015selection}
O.~Tervo, L.-N. Tran, and M.~Juntti, ``{Optimal Energy-Efficient Transmit Beamforming for Multi-User MISO Downlink},'' \emph{IEEE Trans. Signal Process.}, vol.~63, no.~20, pp. 5574--5588, Oct. 2015.

\bibitem{vlachos2020energy}
E.~Vlachos and J.~Thompson, ``{Energy-Efficiency Maximization of Hybrid Massive MIMO Precoding with Random-Resolution DACs via RF Selection},'' \emph{IEEE Trans. Wireless Commun.}, vol.~20, no.~2, pp. 1093--1104, 2020.

\bibitem{choi2021energyIOTJ}
J.~Choi, J.~Park, and N.~Lee, ``{Energy Efficiency Maximization Precoding for Quantized Massive MIMO Systems},'' \emph{IEEE Trans. Wireless Commun.}, Feb. 2021.

\bibitem{zhang2018low}
J.~Zhang, L.~Dai, X.~Li, Y.~Liu, and L.~Hanzo, ``{On Low-Resolution ADCs in Practical 5G Millimeter-Wave Massive MIMO Systems},'' \emph{IEEE Trans. Commun.}, vol.~56, no.~7, pp. 205--211, Apr. 2018.

\bibitem{mao2018rate_bri}
Y.~Mao, B.~Clerckx, and V.~O. Li, ``{Rate-Splitting Multiple Access for Downlink Communication Systems: Bridging, Generalizing, and Outperforming SDMA and NOMA},'' \emph{EURASIP J. Wireless Comm. and Networking}, vol. 2018, no.~1, pp. 1--54, May 2018.

\bibitem{kaleva:tsp:16}
J.~{Kaleva}, A.~Tolli, and M.~{Juntti}, ``{Decentralized Sum Rate Maximization with {QoS} Constraints for Interfering Broadcast Channel Via Successive Convex Approximation},'' \emph{IEEE Trans. Signal Process.}, vol.~64, no.~11, pp. 2788--2802, Jun. 2016.

\bibitem{joudeh:16:tsp}
H.~{Joudeh} and B.~{Clerckx}, ``{Robust Transmission in Downlink Multiuser {MISO} Systems: {A} Rate-Splitting Approach},'' \emph{IEEE Trans. Signal Process.}, vol.~64, no.~23, pp. 6227--6242, Dec. 2016.

\bibitem{dai:twc:16}
M.~{Dai}, B.~{Clerckx}, D.~{Gesbert}, and G.~{Caire}, ``{A Rate Splitting Strategy for Massive {MIMO} with Imperfect {CSIT}},'' \emph{IEEE Trans. Wireless Commun.}, vol.~15, no.~7, pp. 4611--4624, Jul. 2016.

\bibitem{li:jsac:20}
Z.~{Li}, C.~{Ye}, Y.~{Cui}, S.~{Yang}, and S.~{Shamai}, ``{Rate Splitting for Multi-Antenna Downlink: {Precoder} Design and Practical Implementation},'' \emph{IEEE J. Sel. Areas Commun.}, vol.~38, no.~8, pp. 1910--1924, Aug. 2020.

\bibitem{clerckx2016rate}
B.~Clerckx, H.~Joudeh, C.~Hao, M.~Dai, and B.~Rassouli, ``{Rate Splitting for MIMO Wireless Networks: A Promising PHY-Layer Strategy for LTE Evolution},'' \emph{IEEE Commun. Mag.}, vol.~54, no.~5, pp. 98--105, May 2016.

\bibitem{mezghani2012capacity}
A.~Mezghani and J.~A. Nossek, ``{Capacity Lower Bound of MIMO Channels with Output Quantization and Correlated Noise},'' in \emph{Proc. IEEE Int. Symp. Info. Th.}, Jan. 2012, pp. 1--5.

\bibitem{orhan2015low}
O.~Orhan, E.~Erkip, and S.~Rangan, ``{Low Power Analog-to-Digital Conversion in Millimeter Wave Systems: Impact of Resolution and Bandwidth on Performance},'' in \emph{Proc. Inform. Th. and Appl. Workshop}, Feb. 2015, pp. 191--198.

\bibitem{fan2015uplink}
L.~Fan, S.~Jin, C.-K. Wen, and H.~Zhang, ``{Uplink Achievable Rate for Massive MIMO Systems with Low-Resolution ADC},'' \emph{IEEE Commun. Lett.}, vol.~19, no.~12, pp. 2186--2189, Oct. 2015.

\bibitem{jacobsson2017quantized}
S.~Jacobsson, G.~Durisi, M.~Coldrey, T.~Goldstein, and C.~Studer, ``{Quantized Precoding for Massive MU-MIMO},'' \emph{IEEE Trans. Commun.}, vol.~65, no.~11, pp. 4670--4684, Jul. 2017.

\bibitem{zhang2018mixed}
J.~Zhang, L.~Dai, Z.~He, B.~Ai, and O.~A. Dobre, ``{Mixed-ADC/DAC Multipair Massive MIMO Relaying Systems: Performance Analysis and Power Optimization},'' \emph{IEEE Trans. Commun.}, vol.~67, no.~1, pp. 140--153, Sep. 2018.

\bibitem{ribeiro2018energy}
L.~N. Ribeiro, S.~Schwarz, M.~Rupp, and A.~L. de~Almeida, ``{Energy Efficiency of mmWave Massive MIMO Precoding with Low-Resolution DACs},'' \emph{IEEE J. Sel. Topics Signal Process.}, vol.~12, no.~2, pp. 298--312, 2018.

\bibitem{ding2019spectral}
Q.~Ding, Y.~Deng, and X.~Gao, ``{Spectral and Energy Efficiency of Hybrid Precoding for MmWave Massive MIMO with Low-Resolution ADCs/DACs},'' \emph{IEEE Access}, vol.~7, pp. 186\,529--186\,537, 2019.

\bibitem{park2022rate}
S.~Park, J.~Choi, J.~Park, W.~Shin, and B.~Clerckx, ``{Rate-Splitting Multiple Access for Quantized Multiuser MIMO Communications},'' \emph{IEEE Trans. Wireless Commun.}, vol.~22, no.~11, pp. 7696--7711, Nov. 2023.

\bibitem{chen2018alternating}
J.-C. Chen, ``{Alternating Minimization Algorithms for One-Bit Precoding in Massive Multiuser MIMO Systems},'' \emph{IEEE Trans. Veh. Technol.}, vol.~67, no.~8, pp. 7394--7406, 2018.

\bibitem{wang2018finite}
C.-J. Wang, C.-K. Wen, S.~Jin, and S.-H. Tsai, ``{Finite-Alphabet Precoding for Massive MU-MIMO with Low-Resolution DACs},'' \emph{IEEE Trans. Wireless Commun.}, vol.~17, no.~7, pp. 4706--4720, Jul. 2018.

\bibitem{joudeh2016sum}
H.~Joudeh and B.~Clerckx, ``{Sum-Rate Maximization for Linearly Precoded Downlink Multiuser MISO Systems with Partial CSIT: A Rate-Splitting Approach},'' \emph{IEEE Trans. Commun.}, vol.~64, no.~11, pp. 4847--4861, Aug. 2016.

\bibitem{park2021rate}
J.~Park, J.~Choi, N.~Lee, W.~Shin, and H.~V. Poor, ``{Rate-Splitting Multiple Access for Downlink MIMO: A Generalized Power Iteration Approach},'' \emph{IEEE Trans. Wireless Commun.}, Sep. 2022.

\bibitem{fu2020robust}
H.~Fu, S.~Feng, W.~Tang, and D.~W.~K. Ng, ``{Robust secure beamforming design for two-user downlink MISO rate-splitting systems},'' \emph{IEEE Trans. on Wireless Commun.}, vol.~19, no.~12, pp. 8351--8365, 2020.

\bibitem{fletcher2007robust}
A.~K. Fletcher, S.~Rangan, V.~K. Goyal, and K.~Ramchandran, ``{Robust Predictive Quantization: Analysis and Design via Convex Optimization},'' \emph{IEEE J. Sel. Topics Signal Process.}, vol.~1, no.~4, pp. 618--632, Dec. 2007.

\bibitem{gersho2012vector}
A.~Gersho and R.~M. Gray, \emph{{Vector Quantization and Signal Compression}}.\hskip 1em plus 0.5em minus 0.4em\relax Springer, Jun. 2012, vol. 159.

\bibitem{wagner2012large}
S.~Wagner, R.~Couillet, M.~Debbah, and D.~T. Slock, ``{Large system analysis of linear precoding in correlated MISO broadcast channels under limited feedback},'' \emph{IEEE Trans. on Inform. Theory}, vol.~58, no.~7, pp. 4509--4537, 2012.

\bibitem{choi2019GPI}
J.~Choi, N.~Lee, S.-N. Hong, and G.~Caire, ``{Joint User Selection, Power Allocation, and Precoding Design With Imperfect CSIT for Multi-Cell MU-MIMO Downlink Systems},'' \emph{IEEE Trans. Wireless Commun.}, vol.~19, no.~1, pp. 162--176, Sep. 2020.

\bibitem{cui2005energy}
S.~Cui, A.~J. Goldsmith, and A.~Bahai, ``{Energy-Constrained Modulation Optimization},'' \emph{IEEE Trans. Wireless Commun.}, vol.~4, no.~5, pp. 2349--2360, 2005.

\bibitem{hassibi2003much}
B.~Hassibi and B.~M. Hochwald, ``{How much training is needed in multiple-antenna wireless links?}'' \emph{IEEE Trans. on Inform. Theory}, vol.~49, no.~4, pp. 951--963, 2003.

\bibitem{choi2024joint}
J.~Choi, J.~Park, N.~Lee, and A.~Alkhateeb, ``{Joint and Robust Beamforming Framework for Integrated Sensing and Communication Systems},'' \emph{IEEE Trans. on Wireless Commun. (early access)}, 2024.

\bibitem{yoo2006capacity}
T.~Yoo and A.~Goldsmith, ``{Capacity and Power Allocation for Fading MIMO Channels with Channel Estimation Error},'' \emph{IEEE Trans. Inform. Theory}, vol.~52, no.~5, pp. 2203--2214, 2006.

\bibitem{medard2000effect}
M.~Medard, ``{The effect upon channel capacity in wireless communications of perfect and imperfect knowledge of the channel},'' \emph{IEEE Trans. on Inform. theory}, vol.~46, no.~3, pp. 933--946, 2000.

\bibitem{lapidoth2002fading}
A.~Lapidoth and S.~Shamai, ``{Fading channels: how perfect need" perfect side information" be?}'' \emph{IEEE Trans. on Inform. Theory}, vol.~48, no.~5, pp. 1118--1134, 2002.

\bibitem{ding2010maximum}
M.~Ding and S.~D. Blostein, ``{Maximum mutual information design for MIMO systems with imperfect channel knowledge},'' \emph{IEEE Trans. on Inform. Theory}, vol.~56, no.~10, pp. 4793--4801, 2010.

\bibitem{shen2010dual}
C.~Shen and H.~Li, ``{On the Dual Formulation of Boosting Algorithms},'' \emph{IEEE Trans. Pattern Anal. Mach. Intell.}, vol.~32, no.~12, pp. 2216--2231, Mar. 2010.

\bibitem{sriperumbudur2011majorization}
B.~K. Sriperumbudur, D.~A. Torres, and G.~R. Lanckriet, ``{A Majorization-Minimization Approach to the Sparse Generalized Eigenvalue Problem},'' \emph{Machine learning}, vol.~85, no.~1, pp. 3--39, 2011.

\bibitem{cai:siam:18}
Y.~Cai, L.-H. Zhang, Z.~Bai, and R.-C. Li, ``{On an Eigenvector-Dependent Nonlinear Eigenvalue Problem},'' \emph{{SIAM J. Matrix Anal. Appl.}}, vol.~39, no.~3, pp. 1360--1382, Sep. 2018.

\bibitem{armijo1966minimization}
L.~Armijo, ``{Minimization of Functions Having Lipschitz Continuous First Partial Derivatives},'' \emph{Pacific J. Math}, vol.~16, no.~1, pp. 1--3, 1966.

\bibitem{sherman1950adjustment}
J.~Sherman and W.~J. Morrison, ``{Adjustment of an inverse matrix corresponding to a change in one element of a given matrix},'' \emph{Ann. Math. Stat.}, vol.~21, no.~1, pp. 124--127, 1950.

\bibitem{adhi:tit:13}
A.~{Adhikary}, J.~{Nam}, J.~{Ahn}, and G.~{Caire}, ``Joint spatial division and multiplexing - {T}he large-scale array regime,'' \emph{IEEE Trans. Inf. Theory}, vol.~59, no.~10, pp. 6441--6463, Oct. 2013.

\bibitem{erceg1999empirically}
V.~Erceg, L.~J. Greenstein, S.~Y. Tjandra, S.~R. Parkoff, A.~Gupta, B.~Kulic, A.~A. Julius, and R.~Bianchi, ``{An Empirically Based Path Loss Model for Wireless Channels in Suburban Environments},'' \emph{IEEE J. Sel. Areas Commun.}, vol.~17, no.~7, pp. 1205--1211, 1999.

\bibitem{krishnamoorthy2022downlink}
A.~Krishnamoorthy and R.~Schober, ``{Downlink MIMO-RSMA with successive null-space precoding},'' \emph{IEEE Trans. on Wireless Commun.}, vol.~21, no.~11, pp. 9170--9185, 2022.

\bibitem{atzeni2022lowres}
I.~Atzeni, A.~Tölli, and G.~Durisi, ``{Low-Resolution Massive MIMO Under Hardware Power Consumption Constraints},'' in \emph{Asilomar Conference on Signals, Systems, and Computers}, 2021.

\end{thebibliography}

\end{document}